\documentclass[a4,50pt]{amsart}%
\usepackage{amssymb}
\usepackage{pdflscape}
\usepackage{amsmath}
\usepackage{mathrsfs}
\usepackage{amsfonts}
\usepackage{graphicx}%
\usepackage[all]{xy}
\newtheorem{theorem}{Theorem}

\theoremstyle{plain}

\newtheorem{corollary}{Corollary}

\newtheorem*{lemma}{Lemma}

\newtheorem{proposition}{Proposition}

\newtheorem*{remark}{Remark}
\newtheorem*{xrem}{Remark}

\numberwithin{equation}{section}
\begin{document}

\title[\emph{Quantum Dynamics}]{Quantum Dynamics, Minkowski-Hilbert space, and A Quantum Stochastic Duhamel Principle}\author{M.~F.~Brown}
\begin{abstract}
In this paper we shall re-visit the well-known Schr\"odinger and Lindblad dynamics of quantum mechanics. However, these equations shall be realized as the consequence of a more general, underlying dynamical process. In both cases we shall see that the evolution of a quantum state $P_\psi=\varrho(0)$ may be given the not so well-known pseudo-quadratic form $\partial_t\varrho(t)=\mathbf{V}^\star\varrho(t)\mathbf{V}$ where $\mathbf{V}$ is a vector operator in a complex Minkowski space and the pseudo-adjoint $\mathbf{V}^\star$ is induced by the Minkowski metric $\boldsymbol{\eta}$. The interesting thing about this formalism is that its derivation has very deep roots in a new understanding of the differential calculus of time. This Minkowski-Hilbert representation of quantum dynamics is called the \emph{Belavkin Formalism}; a beautiful, but not well understood theory of mathematical physics that understands that both deterministic and stochastic dynamics may be `unraveled' into a second-quantized Minkowski space. Working in such a space provided the author with the means to construct a QS (quantum stochastic) Duhamel principle and simple applications to a Schr\"odinger dynamics perturbed by a continual measurement process are considered. What is not known, but presented here, is the role of the Lorentz transform in quantum measurement and the appearance of Riemannian geometry in quantum measurement is also discussed.
\end{abstract}
\maketitle
%\tableofcontents
\section{Introduction}
Towards the end of the 80's Viacheslav Belavkin had formulated an irreducible matrix representation of the \emph{quantum It\^o algebra} of quantum stochastic increments of an arbitrary quantum system \cite{Be88,Be92,Be92b} which stems from the Fock space representation of stochastic calculus discovered by Hudson and Parthasarathy \cite{HudP84}. Belavkin's formulation was a theory of non-commutative non-adapted stochastic dynamics, later allowing Belavkin and the author to develop a theory of Q-adapted quantum stochastic calculus \cite{BelB12a,BelB12b} recovering the Bose-Fermion equivalence of Hudson and Parthasarathy \cite{HudP84b}. Aside from being a very powerful tool for the study of general stochastic processes the \emph{Belavkin Formalism} has also introduced an entirely new way of understanding differential calculus, investigated  extensively in \cite{Brth} and in particular for deterministic calculus in \cite{Br12}, and its intimate connection with quantum measurement as shall be discussed here.

To fully understand and appreciate the Belavkin Formalism we shall pay a good deal of attention to the simple case of deterministic evolution. In doing this we can arrive at a much deeper understanding of the important physical features of the theory which run the risk of being overlooked in the more general case of stochastic evolution. In fact we shall see that the author's interpretation of the degrees of freedom in the matrix representation of differential increments has some interesting implications about the nature of time. It should, however, be remarked that the author's knowledge of the subject comes from years of personal teachings from Belavkin, the loss of whom is a tragedy.

It is therefore appropriate to begin with the Schr\"odinger equation
\begin{equation}
\mathrm{d}\psi(t)=-\mathrm{i}H\psi(t)\mathrm{d}t,\qquad\psi(0)=\psi\label{11}
\end{equation}
describing, as is well known, the deterministic dynamics of a vector state $\psi$ in a Hilbert space $\mathfrak{h}$; generated of course by a Hamiltonian $H$. %that is self-adjoint on a dense subspace of $\mathfrak{h}$.
What is not well known is that this dynamics has a canonical dilation which comes from understanding $\mathrm{d}t$ as an algebraic object. The deterministic increment $\mathrm{d}t$ is  an element of any  {quantum It\^o algebra} but it may also be regarded as the single basis element of a smaller algebra called \emph{Newton-Leibniz} algebra, a one-dimensional subalgebra of quantum It\^o algebra.
\subsection{Newton-Leibniz Algebra}
Before we can dilate the Schr\"odinger dynamics (\ref{11}) we must pay a little bit of attention to this Newton-Leibniz algebra. It is basically a non-unital algebra $\mathfrak{a}\cong\mathbb{C}\mathrm{d}t$ equipped with the usual algebraic operations of addition and multiplication, but we find that
\begin{equation}
ab=0\;\;\forall\;\;a,b\in\mathfrak{a}\label{12}
\end{equation}
due to the \emph{nilpotent} property of the basis element $\mathrm{d}t$. That means $(\mathrm{d}t)^2=0$. Notice that there is nothing new about this. It is assumed in calculus, where the deterministic derivative of a function $f$ is $\mathrm{d}f(t)=f(t+\mathrm{d}t)-f(t)$, and one can for example verify that if $f(t)=t^n$ then $\mathrm{d}f(t)=nt^{n-1}\mathrm{d}t$ requires that $(\mathrm{d}t)^2=0$. In addition to being nilpotent $\mathfrak{a}$ is also equipped with an involution
\begin{equation}
\star:a=\alpha\mathrm{d}t\mapsto{\alpha}^\ast\mathrm{d}t=a^\star\label{13}
\end{equation}
where $\alpha^\ast$ is  the complex conjugation of the coefficient $\alpha$ in $\mathbb{C}$.
In addition to the Newton-Leibniz algebra being a nilpotent algebra with involution it is also equipped with a $\mathbb{C}$-linear functional
$l:\mathfrak{a}\rightarrow\mathbb{C}$
satisfying the conditions
\begin{equation}
l(\mathrm{d}t)=1,\qquad l(a^\star)=l(a)^\ast,\label{14}
\end{equation}
but it may only be called \emph{pseudo-positive} due to the nilpotent property of $\mathfrak{a}$, such that $l(a^\star a)=0$ for all $a\in\mathfrak{a}$ and not just for $a=0$. For this reason we may only refer to $l$ as a \emph{pseudo}-state.

The following theorem is a special case of that given by Belavkin in \cite{Be92b}, but it is an important thing to understand  in its own right. Serving also as a preliminary to understanding Belavkin's more general theorem the notations used here  are in accordance with his own, and we shall also see that this is similar in spirit to the GNS construction representing $C^\ast$-algebras as operator sub-algebras in a Hilbert space.
\begin{theorem}
Let $(\mathfrak{a},l)$ define a Newton-Leibniz algebra. Then, up to isomorphism, there is a unique irreducible matrix representation  of $\mathfrak{a}$, denoted $\pi(\mathfrak{a})$, on a pseudo-Hilbert space $\Bbbk\cong\mathbb{C}^2$ equipped with the Minkowski metric $\boldsymbol{\eta}$, and there is a null vector $\xi\in\Bbbk$, $\xi^\star\xi=0$, such that $\xi^\star\pi(a)\xi=l(a)$, where $\xi^\star=\xi^\ast\boldsymbol{\eta}$ and $\xi^\ast$ is the standard conjugation of row-column transposition and complex conjugation.
\end{theorem}
\begin{proof}
The basic ingredients of this proof are as follows. First we shall add a unit $\boldsymbol{1}$ to the algebra $\mathfrak{a}$ to form the monoid $\mathfrak{a}_{\boldsymbol{1}}$ where elements of this group have the form $\boldsymbol{1}+a$ and the action of $\boldsymbol{1}$ is defined on $\mathfrak{a}$ as $\boldsymbol{1}a=a=a\boldsymbol{1}$ for all $a\in\mathfrak{a}$. The pseudo-state $l$ is also extended onto the unit by the definition $l(\boldsymbol{1})=0$, so that $l(\boldsymbol{1}+a)=l(a)$.

The $\star$-involution is also extended onto the unit so that $\boldsymbol{1}^\star=\boldsymbol{1}$ and for elements $f$ and $g$ in $\mathfrak{a}_{\boldsymbol{1}}$ we define a sesquilinear form %linear in one argument anti-linear in other
\[
\langle f,g\rangle:=l(f g^\star).
\]
Next we proceed with the introduction of the Kronecker delta $\delta_{g,f}$ and introduce a space $\mathfrak{B}$ of maps $\kappa:\mathfrak{a}_{\boldsymbol{1}}\rightarrow\mathbb{C}$ having finite support, and thus given by the finite sums $\kappa=\sum_{g\in\mathfrak{a}_{\boldsymbol{1}}}\kappa_g\delta_g$ where $\kappa_g\in\mathbb{C}$ is non-zero for only a finite set of elements $g\in\mathfrak{a}_{\boldsymbol{1}}$, and $\delta_g=(\delta_{g,f})_{f\in\mathfrak{a}_{\boldsymbol{1}}}$. This space is simply a tool for representing $\mathfrak{a}_{\boldsymbol{1}}$, which is achieved by defining the algebraic operation $\varepsilon\star\kappa$ as
\[
(\varepsilon\star\kappa)_f=\sum_{gh^\star=f}\varepsilon_g\kappa_h^\ast,\quad f,g,h\in\mathfrak{a}_{\boldsymbol{1}}
\]
 which is well-defined in $\mathfrak{B}$ for all $\varepsilon, \kappa\in\mathfrak{B}$, so that we have a $\star$-representation of $\mathfrak{a}_{\boldsymbol{1}}$ in $\mathfrak{B}$ given by $\delta_f\star\delta_g=\delta_{fg^\star}$.
 In terms of these functions $\delta_g$ we may write the sesquilinear form as
\[
(\delta_f|\delta_g):= \langle f,g\rangle\equiv l(f g^\star)
\]
and then form the equivalence class $\mathfrak{B}/\mathfrak{K}^\star$ where $\mathfrak{K}^\star=\big\{\kappa^\star:\kappa\in\mathfrak{K}\big\}$ for $\kappa^\star:=\sum_{g\in\mathfrak{a}_{\boldsymbol{1}}}\kappa^\ast_{g^\star} \delta_g$ and
 \[
 \mathfrak{K}=\Big\{\varepsilon:(\varepsilon|\kappa)\equiv\sum_{f,g}\varepsilon_f(\delta_f|\delta_g) \kappa_g^\ast=0\;\forall\;\kappa\in\mathfrak{B}\Big\},
 \]
 which is the space of functions $\kappa$ for which $\sum_g\kappa_g=0$ and $\sum_g\kappa_g l(g)=0$.
The equivalence class $\mathfrak{B}/\mathfrak{K}^\star$ is a vector space $\Bbbk$ of column vectors $|\kappa)$ which may be given the form
\begin{equation}
|\kappa)=\left[
           \begin{array}{c}
             \kappa_+^\ast \\
             \kappa_-^\ast \\
           \end{array}
         \right]:=\left[
           \begin{array}{c}
             \kappa^- \\
             \kappa^+ \\
           \end{array}
         \right],\;\;\textrm{where}\; \kappa_-:=\sum_g\kappa_g\;\;\textrm{and}\;\kappa_+:=\sum_g\kappa_gl(g).\label{-1}
\end{equation}
 Then to obtain the Minkowski metric one should note that $l(fg^\star)=l(f)+l(g^\star)$ so that the indefinite inner product on $\Bbbk$ has the form
\begin{equation}
(\varepsilon|\kappa)=\sum_{f,g}\varepsilon_fl(f)\kappa_g^\ast+\sum_{f,g}\varepsilon_fl(g^\star)\kappa_g^\ast \equiv \varepsilon_+\kappa^++\varepsilon_-\kappa^-.\label{-2}
\end{equation}
But also note that the Minkowski metric does not have its usual diagonal form here, but instead we work with an isomorphism of the usual diagonal form, namely
\begin{equation}
\boldsymbol{\eta}=\left[
         \begin{array}{cc}
           0 & 1 \\
           1 & 0 \\
         \end{array}
       \right].\label{eta}
\end{equation}
 Finally we may define the representation $\pi(\mathfrak{a}_{\boldsymbol{1}})$ on the vectors $|\delta_g)$ as
\[
\pi(f)|\delta_g)=|\delta_{gf^\star})\quad\textrm{with}\quad |\delta_{\boldsymbol{1}})=\xi %\qquad (\delta_g|\pi(f)=(\delta_{gf^\star}|
\]
having the explicit form
\begin{equation}
\pi(g)=\left[
         \begin{array}{cc}
           1 & l(g) \\
           0 & 1 \\
         \end{array}
       \right]\equiv\mathbf{I}+l(g)\pi(\mathrm{d}t),\qquad \xi=\left[
                            \begin{array}{c}
                              0 \\
                              1 \\
                            \end{array}
                          \right],\label{xi}
\end{equation}
with the adjoint induced by the Minkowski metric $\boldsymbol{\eta}$ as $\pi(g^\star)=\boldsymbol{\eta}^{-1}\pi(g)^\ast\boldsymbol{\eta}=\pi(g)^\star$, where the $\ast$-involution is the usual matrix involution of column-row transposition and complex conjugation. Indeed $\xi^\star\pi(g)\xi=l(g)$.
\end{proof}
\subsection{Canonical Dilation of Schr\"odinger Dynamics}
Having established this canonical representation of Newton-Leibniz algebra, which underlies all deterministic dynamics, we may proceed with the dilation of the Schr\"odinger equation (\ref{11}). This gives us an equation having the form
\begin{equation}
\mathrm{d}\psi_t(\vartheta)=\mathbf{L}\psi_t(\vartheta)\mathrm{d}n_t(\vartheta),\qquad \psi_0(\vartheta)=\psi\otimes\xi(\vartheta)\label{15}
\end{equation}
where $\vartheta$ is an ordered set of $n=|\vartheta|$ points $\{t_1<\ldots<t_n\}$ and $\mathrm{d}n_t(\vartheta):=1$ if $t\in\vartheta$ and zero otherwise. Equation (\ref{15}) describes a discrete interaction dynamics, a sequence of interactions at points $t\in\vartheta$, between the quantum system $\psi$ and a canonical auxiliary system $\xi(\vartheta)$. Here, $\xi(\vartheta)$ is a product of the vectors $\xi(t)$, $t\in\vartheta$, given by (\ref{xi}), and $\mathbf{L}$ mediates the interactions between $\psi$ and each $\xi(t)$. The spontaneous interactions have the explicit form
\begin{equation}
\mathbf{L}:=-\mathrm{i}H\otimes\pi(\mathrm{dt}),\qquad \pi(\mathrm{d}t)=\left[
                   \begin{array}{cc}
                     0 & 1 \\
                     0 & 0 \\
                   \end{array}
                 \right]\label{16}
\end{equation}
which is the canonical matrix representation of $-\mathrm{i}H\mathrm{d}t$ from (\ref{11}) and the appearance of $\mathrm{d}n_t(\vartheta)$ in (\ref{15}) is simply to indicate that there is change in the quantum system if and only if it interacts with the auxiliary system $\xi(\vartheta)$. Now we must understand what this auxiliary system is.

From the point of view of physics the vector $\xi$ is a very specific thing. It holds the information about whether an event, described by an interaction, has taken place or not and underlies Belavkin's \emph{Eventum Mechanics} \cite{Be00b}. To illustrate this we may consider an initial state $\psi\otimes\xi$ that represents a quantum system in a state $\psi$ composed with a fundamental `event triggering' wave-function $\xi$. Prior to any interaction we may identify $\xi$ with the statement: \emph{the event is in the future}, where \emph{future} simply refers to the domain of that which has not happened - potential. Now we shall consider the change $\mathbf{L}(\psi\otimes\xi)$ arising from the interaction between $\psi$ and $\xi$, which operates as
\begin{equation}
\mathbf{L}:\psi\otimes\left[
                        \begin{array}{c}
                          0 \\
                          1 \\
                        \end{array}
                      \right]\mapsto-\mathrm{i}H\psi\otimes \left[
                        \begin{array}{c}
                          1 \\
                          0 \\
                        \end{array}
                      \right].\label{17}
\end{equation}
Notice that this not only generates a differential increment of the quantum system, $\psi\mapsto\dot{\psi}:=-\mathrm{i}H\psi$, but the vector $\xi$ is also transformed.  It is this transformation that serves as an indication that the interaction has taken place and at the most basic level we have
\begin{equation}
\pi(\mathrm{d}t)
:\left[
                        \begin{array}{c}
                          0 \\
                          1 \\
                        \end{array}
                      \right]\mapsto\left[
                        \begin{array}{c}
                          1 \\
                          0 \\
                        \end{array}
                      \right].\label{18}
\end{equation}
The transformation  $\pi(\mathrm{d}t)\xi$ corresponds to the complimentary statement: \emph{the event is in the past}, where \emph{past} refers to the domain of that which has happened - actual, and (\ref{18}) describes an underlying mechanism for generating past from future. The interaction operator generating the change (\ref{17}) is $\mathbf{G}=\mathbf{I}+\mathbf{L}$, see (\ref{xi}). It is $\star$-unitary such that $\mathbf{G}^{-1}=\mathbf{G}^\star$ and its action on $\psi\otimes\xi$ may be understood as entangling the quantum system with the \emph{temporal spin} $\xi$.  The initial state $\psi\otimes\xi(t_1)\otimes\cdots\otimes\xi(t_n)$ for the dynamics (\ref{15}) is a composition of the quantum system with $n$ copies of the temporal spin vector $\xi$ in the future pseudo-state and each  interaction that occurs generates a past component containing a derivative of the quantum system. The sequential process generated by this mechanism gives rise to the sequence of differential increments
\begin{equation}
\psi\mapsto\dot{\psi}\mapsto\ddot{\psi}\mapsto\ldots.\label{19}
\end{equation}
Indeed this has been noted by Belavkin and may also have been recognized by Kolokoltsov \cite{BelK01}.

Some may feel that it is unnecessary to complicate the Schr\"odinger dynamics in such a way, yet this `pseudo-Stinespring' dilation of the Schr\"odinger dynamics reveals foundations of our conception of time and evolution. In fact, as far as mathematics is concerned, we can now say that such deterministic evolution is generated by spontaneous interactions at a boundary between past and future \cite{Be00b,Be00a}. But here it is emphasized that past and future are now regarded as two fundamental components of the mechanism that generates evolution, a fundamental two-point quantization of time, $\{-,+\}$ \cite{Brth,Br12}.

The order on the  coordinates in $\vartheta$ may be regarded as ultimately induced by the temporal spin in so far as one event may be called less-than another if it has happened and the other has not. We consider $\vartheta$ as a random variable indicating potential interactions given as any finite set of points in a connected subset $\Delta\subset\mathbb{R}$ equipped with the total order induced from $\mathbb{R}$. To accommodate this we denote by $\boldsymbol{F}^\star$ the embedding
\begin{equation}
\boldsymbol{F}^\star:\psi\mapsto\psi\otimes\xi^\otimes\label{110}
\end{equation}
of $\psi\in\mathfrak{h}$ into $\mathbb{H}=\mathfrak{h}\otimes\mathbb{F}$ where $\mathfrak{h}$ is an arbitrary Hilbert space for a quantum system described by $\psi$ and $\mathbb{F}=\Gamma(\Bbbk)$ is the pseudo-Fock space over the pseudo-Hilbert space $\Bbbk$
of square-integrable vector functions $\xi:\Delta\rightarrow (\mathbb{C}^2,\boldsymbol{\eta})$, called Minkowski-Hilbert space \cite{Br12}. By $\xi^\otimes$ we are  denoting the maps
\begin{equation}
\xi^\otimes:\vartheta\mapsto\xi^\otimes(\vartheta)=\underset{{t\in\vartheta}}{\otimes}\xi(t)\label{111}
\end{equation}
which was short-handedly  written as $\xi(\vartheta)$ in (\ref{15}), but $\xi^\otimes$ is more concise. More details about Minkowski-Fock space may be found in the appendix.

The solution of (\ref{15}) has the general form
\begin{equation}
\psi_t=\mathbf{G}^\odot_t\boldsymbol{F}^\star\psi,\qquad \mathbf{G}^\odot_t(\vartheta):=\overset{\leftarrow}{\underset{x\in\vartheta\cap[0,t)}{\odot}}\mathbf{G}(t),\label{112}
\end{equation}
where $t\in\Delta$ and $\min\Delta=0$ has been taken for convenience, and we also suppose that $\mathbf{G}(t)=\mathbf{G}$ is not depending on time, for simplicity. The product $\odot$ is a semi-tensor product which preserves the chronology of the interactions \cite{Be88} and may be defined as
\begin{equation}
\mathbf{G}(y)\odot\mathbf{G}(x):=\Big(\mathbf{G}(y)\otimes\mathbf{I}(x)\Big)\Big(\mathbf{I}(y)\otimes\mathbf{G}(x)\Big), \quad y>x. \label{113}
\end{equation}
The operator $\boldsymbol{F}^\star$ is an isometry such that $\boldsymbol{F}\boldsymbol{F}^\star=I$ in $\mathfrak{h}$, corresponding to the normalization of $\xi^\otimes$ given by $\|\xi^\otimes\|^2:=\exp\{\xi^\star\xi\}=1$ (see appendix), and $\mathbf{E}=\boldsymbol{F}^\star\boldsymbol{F}$ defines a projection of $\mathbb{H}$ into $\mathfrak{h}$. The solution $U(t)=\exp\{-\mathrm{i}Ht\}$ of the Schr\"odinger equation (\ref{11}) is obtained as the projection $U(t)=\boldsymbol{F}\mathbf{G}^\odot_t\boldsymbol{F}^\star$ which is a conditional expectation of the sequential interaction dynamics \cite{Br12}, and notice that  $\mathbf{E}\mathbf{G}^\odot_t\boldsymbol{F}^\star=U(t)\boldsymbol{F}^\star$ is indeed the projection of the sequential interaction dynamics into the Hilbert space $\mathfrak{h}$ of the quantum system.

The Schr\"odinger equation for the density operator $\varrho=\psi\psi^\dag\equiv P_\psi$ has the familiar form
\begin{equation}
\mathrm{d}\varrho(t)=\mathrm{i}[\varrho(t),H]\mathrm{d}t
\end{equation}
but now that we have established the form of the canonical dilation we may write this in the pseudo-quadratic form
\begin{equation}
\partial_t\varrho(t)=\mathbf{V}^\star\big(\varrho(t)\otimes\mathbf{I}\big)\mathbf{V}\equiv \mathbf{V}^\star\varrho(t)\mathbf{V},\label{gend}
\end{equation}
as all time-dependence is explicit,
where $\mathbf{V}^\star=\xi^\star\mathbf{G}\equiv (I,-\mathrm{i}H)$ is a row-vector operator in $\mathcal{B}(\mathfrak{h})\otimes\Bbbk(t)$, $\mathbf{V}=\mathbf{G}^\star\xi$ is the pseudo-adjoint column, and $\mathbf{V}^\star\mathbf{V}=0$. $\mathcal{B}(\mathfrak{h})$ is the algebra of bounded operators in $\mathfrak{h}$
\begin{remark}
It is worth noting that $\Bbbk$ is not $L^2(\Delta)\otimes\mathbb{C}^2$ but is in fact $L^1(\Delta)\oplus L^\infty(\Delta)$, which is indeed pseudo-Hilbert space with respect to the Minkowski metric $\boldsymbol{\eta}$ \cite{Be92b,Br12}.
\end{remark}
\section{Quantum Stochastic Differential Equations}
The vector operator $\mathbf{V}^\star=(I,-\mathrm{i}H)$ may be regarded as a special case of $\mathbf{V}^\star=(I,K)$ where $K$ need no longer be anti-selfadjoint. However, with respect to the indices $\{-,+\}$ of the degrees of freedom in the Minkowski-Hilbert space $\Bbbk$ one may write $\mathbf{V}^\star=(V_-,V_+)$. Then the next question that arises is what happens when we introduce additional degrees of freedom other than these fundamental two? For example, let's consider a Minkowski-Hilbert space $\mathbb{K}\cong(\mathbb{C}^{n+2},\boldsymbol{\eta})$ where the Minkowski metric now has the form
\begin{equation}
 \boldsymbol{\eta}=\left[
                     \begin{array}{ccc}
                       0 & 0 & 1 \\
                       0 & \mathrm{I} & 0 \\
                       1 & 0 & 0 \\
                     \end{array}
                   \right]\label{2004}
 \end{equation}
and $\mathrm{I}$ is the $n$-dimensional identity operator on $\mathbb{C}^n$. Now we shall consider a dynamical equation of the form
(\ref{gend}) but this time with $\mathbf{V}^\star=(V_-,V_1,\ldots,V_n,V_+)\equiv(V_-,\mathrm{V}_\bullet,V_+)$, where $\mathrm{V}_\bullet=(V_1,\ldots,V_n)$ and $\bullet=\{1,2,\ldots,n\}$ is just a suppression of the additional degrees of freedom. We then find that
\begin{equation}
\partial_t\varrho(t)= V_-\varrho(t)V_+^\dag+\mathrm{V}_\bullet\varrho(t){\mathrm{V}_\bullet}^\ast+V_+\varrho(t)V_-^\dag \equiv \mathbf{V}^\star\varrho(t)\mathbf{V} ,\label{lindblad}
\end{equation}
where $\mathrm{V}_\bullet\varrho(t){\mathrm{V}_\bullet}^\ast= \sum_{i=1}^nV_i\varrho(t)V_i^\dag$ and ${\mathrm{V}_\bullet}^\ast$ is the column vector of operators $[V_i^\dag]_{i\in\bullet}$ which is not the same is the row vector of operators $\mathrm{V}_\bullet^\dag=(V_1^\dag,\ldots,V_n^\dag)$.

Again we shall take $V_-=I$ and we shall also write $V_+=K$, and in addition to this we shall denote $\mathrm{V}_\bullet$ by the symbol $\widetilde{\mathrm{L}}:=\mathrm{L}^{\ast\dag}$. Then (\ref{lindblad}) assumes the form
\begin{equation}
\partial_t\varrho(t)= \varrho(t)K^\dag+\widetilde{\mathrm{L}}\varrho(t)\widetilde{\mathrm{L}}^\ast+K\varrho(t) ,\label{lindblad2}
\end{equation}
which is the Lindblad equation when $K^\dag=\mathrm{i}H-\frac{1}{2}\mathrm{L}^\ast\mathrm{L}$ on noticing that $\varrho(t)K^\dag+K\varrho(t)= \mathrm{i}\big[\varrho(t),H]-\big\{\varrho(t),\tfrac{1}{2}\mathrm{L}^\ast\mathrm{L}\big\}$,
where $[A,B]=AB-BA$ is the usual commutator and $\{A,B\}=AB+BA$ is the anti-commutator.
\begin{remark}Although it may look like a typing error it is convenient to use Latin letters $B$ for operators on the Hilbert space $\mathfrak{h}$ and Roman letters $\mathrm{B}$ for operators affiliated to another Hilbert space $\mathfrak{k}$ that shall ultimately be used for studying the noise - if $\mathfrak{k}$ is $n$-dimensional it is because the noise has $n$-degrees of freedom. In addition to this it proves useful to denote involution of $\mathfrak{h}$ by $\dag$ but have an involution $\ast$ of the product space $\mathfrak{h}\otimes\mathfrak{k}$. The reason for this is that we would like to decompose the Hermitian involution $\mathrm{X}$ of an operator on $\mathfrak{h}\otimes\mathfrak{k}$ as $\mathrm{X}^\ast=\widetilde{\mathrm{X}}^\dag$, where $\widetilde{\mathrm{X}}$ is an unusual involution that is a transposition of column and row structure in $\mathcal{B}(\mathfrak{k})$ but it does not transpose the column and row structure in $\mathcal{B}(\mathfrak{h})$. This will appear naturally in what follows, but one may wish to return to the above and observe the prior notations of this section in light of this.
 \end{remark}

 It is all well and good writing  down the Lindblad equation but now we must derive all this from some kind of underlying process as has been done in the case of the Schr\"odinger dynamics above. Indeed this means understanding the origin of such additional degrees of freedom collectively indexed by $\bullet$.

 \emph{Theorem 1} has, as mentioned, a more general counterpart in \cite{Be92b} where the canonical matrix representation of the remaining three quantum stochastic differential increments was constructed. We will not go into the details of this here, but what we will say is this. The dilation of deterministic dynamics gives rise to two irreducible degrees of freedom at the differential level, $\{-,+\}$, as we have already seen, but if there is also stochastic dynamics generated by noise then we have at least three irreducible degrees of freedom: the past and future degrees of freedom of deterministic time and at least one additional degree of freedom for the noise, $\{-,\bullet,+\}$.

\subsection{Quantum It\^o Algebra}
This is an algebra
\begin{equation}
\mathfrak{a}(\mathfrak{k})=\mathbb{C}\mathrm{d}t+\mathfrak{k}^\ast\otimes\mathrm{d}a+ \mathcal{B}(\mathfrak{k})\otimes\mathrm{d}n +\mathfrak{k}\otimes\mathrm{d}a^\star,
\end{equation}
where $\mathcal{B}(\mathfrak{k})$ is the space of bounded operators on a Hilbert space $\mathfrak{k}$. The basis elements $\{\mathrm{d}t,\mathrm{d}a,\mathrm{d}n, \mathrm{d}a^\star\}$, with $\mathrm{d}t^\star=\mathrm{d}t$, $\mathrm{d}n^\star=\mathrm{d}n$, and $\mathrm{d}a^\star\neq\mathrm{d}a$, are not orthogonal but satisfy the H-P multiplication table \cite{HudP84}
\begin{equation}
\mathrm{d}a\mathrm{d}n=\mathrm{d}a,\quad \mathrm{d}n\mathrm{d}n=\mathrm{d}n,\quad \mathrm{d}n\mathrm{d}a^\star=\mathrm{d}a^\star,\quad \mathrm{d}a\mathrm{d}a^\star=\mathrm{d}t,\label{mt}
\end{equation}
with all other products being zero. Notice that the Newton-Leibniz algebra is indeed a subalgebra $\mathbb{C}\mathrm{d}t\subset \mathfrak{a}(\mathfrak{k})$.
If the Hilbert space $\mathfrak{k}$ is separable then we may write the elements $a\in\mathfrak{a}(\mathfrak{k})$ in the form
\begin{equation}
a=\mathrm{D}^-_+\mathrm{d}t+\mathrm{D}^-_i\mathrm{dA}^i+\mathrm{D}^k_i\mathrm{dN}^i_k+\mathrm{D}^k_+\mathrm{dA}^\star_k \equiv \mathrm{D}^\kappa_\iota\mathrm{d}\Lambda^\iota_\kappa,\label{element}
\end{equation}
where $\kappa\in\{\bullet,+\}$ and $\iota\in\{-,\bullet\}$
and we have assumed the Einstein summation convention, and H-P multiplication table for component increments is
\begin{equation}
\mathrm{dA}^i\mathrm{dN}^j_k=\delta^i_k\mathrm{dA}^j,\;\; \mathrm{dN}^i_l\mathrm{dN}^j_k=\delta^i_k\mathrm{dN}^j_l,\;\; \mathrm{dN}^i_l\mathrm{dA}_k^\star=\delta^i_k\mathrm{dA}_l^\star,\;\;\mathrm{dA}^i\mathrm{dA}_k^\star= \delta^i_k\mathrm{d}t.\label{mt2}
\end{equation}
The right-hand side of (\ref{element}) is given with respect to the basis of differential increments in Belavkin's notation
\begin{equation}
\mathrm{d}\Lambda^+_-\equiv\mathrm{d}t,\quad\mathrm{d}\Lambda^i_-\equiv\mathrm{dA}^i, \quad\mathrm{d}\Lambda^i_k\equiv\mathrm{dN}^i_k,\quad\mathrm{d}\Lambda^+_k\equiv\mathrm{dA}_k^\star.\label{Be}
\end{equation}
of the Belavkin notation, which satisfy the H-P multiplication table in the form
\begin{equation}
\mathrm{d}\Lambda^\alpha_\kappa \mathrm{d}\Lambda^\iota_\beta=\delta^\alpha_\beta\mathrm{d}\Lambda^\iota_\kappa
\end{equation}
and we may give the product of two elements $a=\mathrm{D}^\kappa_\iota\mathrm{d}\Lambda^\iota_\kappa$ and $b=\mathrm{F}^\kappa_\iota\mathrm{d}\Lambda^\iota_\kappa$ as
\begin{equation}
ab=\mathrm{D}^\kappa_\alpha\mathrm{d}\Lambda^\alpha_\kappa\mathrm{F}_\iota^\beta\mathrm{d}\Lambda_\beta^\iota= \mathrm{D}^\kappa_\alpha\mathrm{F}_\iota^\alpha\mathrm{d}\Lambda_\iota^\kappa \label{prod}
\end{equation}
generalizing the nilpotenet product (\ref{12}) of the deterministic differential calculus. Notice that the product $\mathrm{D}^\kappa_\alpha\mathrm{F}_\iota^\alpha$ in (\ref{prod}) has the appearance of a matrix product $(\mathbf{DF})^\kappa_\iota$. This is no accident, the index notation used by Belavkin has its origins in the differential-calculus matrix-representation which he constructed.

We have already seen the canonical matrix representation for the increment $\mathrm{d}t$, and here it is the same but given in at least one extra dimension as
\[
\pi(\mathrm{d}t)=\left[
                   \begin{array}{ccc}
                     0 & 0 & 1 \\
                     0 & 0 & 0 \\
                     0 & 0 & 0 \\
                   \end{array}
                 \right]=\pi(\mathrm{d}\Lambda^+_-).
\]
The remaining increments of annihilation, counting and creation are respectively given as
\begin{equation}
\pi(\mathrm{dA}^k)=\left[
                   \begin{array}{ccc}
                     0 & \langle k| & 0 \\
                     0 & 0 & 0 \\
                     0 & 0 & 0 \\
                   \end{array}
                 \right],\; \pi(\mathrm{dN}^k_i)=\left[
                   \begin{array}{ccc}
                     0 & 0 & 0 \\
                     0 & |i\rangle\langle k| & 0 \\
                     0 & 0 & 0 \\
                   \end{array}
                 \right],\; \pi(\mathrm{dA}_i^\star)=\left[
                   \begin{array}{ccc}
                     0 & 0 & 0 \\
                     0 & 0 & |i\rangle \\
                     0 & 0 & 0 \\
                   \end{array}
                 \right],\nonumber
\end{equation}
with respect to the vectors $\langle k|=(\delta^k_i)_{i\in\bullet}=|k\rangle^\ast$.
In sight of (\ref{Be})   we may write a general QS (quantum stochastic) differential $a=\mathrm{L}^\kappa_\iota\mathrm{d}\Lambda^\iota_\kappa$  in its canonical matrix form, as
\begin{equation}
\pi(a)=\left[
                   \begin{array}{ccc}
                     0 & \mathrm{L}^-_\bullet & \mathrm{L}^-_+ \\
                     0 & \mathrm{L}^\bullet_\bullet & \mathrm{L}^\bullet_+ \\
                     0 & 0 & 0 \\
                   \end{array}
                 \right]\equiv\mathbf{L}=\mathrm{L}^\kappa_\iota\otimes\pi(\mathrm{d}\Lambda^\iota_\kappa)
\end{equation}
generalizing (\ref{16}), and represented on the space $\mathbb{K}(t)=\mathbb{C}\oplus\mathfrak{k}\oplus\mathbb{C}$ equipped with the Minkowski metric (\ref{2004}) with $\mathrm{I}$ as the identity on  $\mathfrak{k}$, whose degrees of freedom are represented by $\bullet$.  It is useful to consider the degrees of freedom of the noise as the \emph{internal} degrees of freedom of the \emph{external} noise index $\bullet$, particularly when discussing the \emph{external} upper-triangular structure of QS differential increments, which is preserved by the involution $\mathbf{L}^\star=\boldsymbol{\eta}^{-1}\mathbf{L}^\ast\boldsymbol{\eta}$.

The pseudo-state $l$ given in (\ref{14}) is given more generally on the unitilization $\mathfrak{a}_{\boldsymbol{1}}(\mathfrak{k})$ of the quantum It\^o algebra $\mathfrak{a}(\mathfrak{k})$, adding the matrix identity $\mathbf{I}$ to the representations $\pi(a)$ of the elements $a$. Although its action on the basis elements $\mathrm{d}\Lambda^\kappa_\iota$ is zero for all but $\mathrm{d}\Lambda^+_-=\mathrm{d}t$ it does give us an additional term on the product of two elements $f,g\in \mathfrak{a}_{\boldsymbol{1}}(\mathfrak{k})$ as
\begin{equation}
l(fg^\star)=l(f)+k(f)k(g)^\ast+l(g)^\ast
\end{equation}
where $k(f)$ is a row-vector, $\mathrm{D}_\bullet^-$ say, arising from the presence of noise. The matrix representation of $\mathfrak{a}_{\boldsymbol{1}}(\mathfrak{k})$ is constructed by imposing on $l$ a condition for positivity, called \emph{conditional positivity} in \cite{Be92b},  and the vector $\xi$ appearing in the presence of noise now has the external form
\begin{equation}
\xi=(1,0,0)^\star\label{xinoise}
\end{equation}
with respect to the three {external} degrees of freedom $\{-,\bullet,+\}$ of a QS differential increment, and the pseudo-state $l$ gives us the deterministic derivative of the dynamics as $l\big(\mathrm{L}^\kappa_\iota\mathrm{d}\Lambda^\iota_\kappa\big)=\mathrm{L}^-_+=\xi^\star\mathbf{L}\xi$.

\subsection{From Fock Space Operators to Block-Diagonal Minkowski-Fock Space Operators}
We have already seen that Schr\"odinger dynamics dilates to  a discrete interaction dynamics given by (\ref{15}) but notice that the propagator $\mathbf{U}_t=\mathbf{G}_t^\odot$ in (\ref{112}) is a block-diagonal operator on $\mathbb{H}$. This means that $[\mathbf{U}_t\psi_0](\vartheta)=\mathbf{U}_t(\vartheta)\psi_0(\vartheta)$. In addition to this the blocks $\mathbf{U}_t(\vartheta)$ are upper-triangular. In the case when $\mathfrak{h}=\mathbb{C}$  the operators $\mathbf{U}_t(\vartheta)$ are elements of the $\star$-monoid $\mathbb{M}(\vartheta)=\mathfrak{a}_{\boldsymbol{1}}(t_1)\otimes\cdots\otimes\mathfrak{a}_{\boldsymbol{1}}(t_n)$, $\vartheta=\{t_1<\ldots<t_n\}$, where each $\mathfrak{a}_{\boldsymbol{1}}(t)$ is a copy of the unitalization $\mathfrak{a}_{\boldsymbol{1}}$ of the Newton-Leibniz algebra $\mathfrak{a}$, and we denote by $\mathbb{M}=\mathfrak{a}_{\boldsymbol{1}}^\otimes$ the  block-diagonal upper-triangular $\star$-monoid containing the operators $\mathbf{U}_t$. More generally $\mathfrak{h}\neq\mathbb{C}$ and we define $\mathbb{M}=\mathfrak{m}^\odot$ where $\mathfrak{m}=\boldsymbol{1}+\mathfrak{a}\otimes\mathcal{B}(h)$ which is an affine space of more general operators  $\mathbf{X}$ having the components ${X}^{\iota_1\;\iota_2\ldots\;\iota_n}_{\kappa_1\kappa_2\ldots\kappa_n}(\vartheta)\in\mathcal{B}(\mathfrak{h})$ which are zero for any $(\iota_i,\kappa_i)=(+,-)$ and if $\iota_i=\kappa_i$ the operator $\mathbf{X}(\vartheta)$ acts as the identity at the point $t_i\in\vartheta$.

When $\mathbb{H}=\mathfrak{h}\otimes\mathbb{F}$ is given by the
Bigger Minkowski-Fock space $\mathbb{F}=\Gamma(\mathbb{K})$ having the noise degrees of freedom $\bullet$ so that $\mathbb{K}(t)=\mathbb{C}\oplus\mathfrak{k}\oplus\mathbb{C}$  we may study dilations (or in fact matrix representations) of general quantum stochastic dynamics. Again we consider a block-diagonal (externally) upper-triangular monoid $\mathbb{M}$, only this time it is given by the second quantization $\mathbb{M}=\mathfrak{m}^\odot$ of the monoid $\mathfrak{m}=\boldsymbol{1}+\mathfrak{a}(\mathfrak{k})\otimes\mathcal{B}(h)$, where $\mathfrak{a}(\mathfrak{k})$ is a general quantum It\^o algebra over the Hilbert space $\mathfrak{k}$. To understand the form of general operators $\mathbf{X}$ in $\mathbb{M}$ it will prove useful to introduce a very important operator called the \emph{point-derivative} $\boldsymbol{\nabla}:\mathbb{M}\rightarrow\mathfrak{m}\odot\mathbb{M}:=\dot{\mathbb{M}}$ and its action is defined on the operators $\mathbb{X}$ as
\begin{equation}
\big[\boldsymbol{\nabla}(\mathbf{X})\big](t,\vartheta)= \mathbf{X}(t\sqcup\vartheta)\equiv\dot{\mathbf{X}}(t,\vartheta)\label{pd}
\end{equation}
resembling the Malliavin derivative \cite{Mal78}where $\sqcup$ denotes the disjoint union, $t\notin\vartheta$.
 Now the operators $\mathbf{X}$ may be characterized by the formula
 \begin{equation}
 \dot{\mathbf{X}}(t,\vartheta\setminus t)=\left[
                                                                   \begin{array}{ccc}
                                                                    \mathbf{X}(\vartheta\setminus t) & \dot{\mathbf{X}}^-_\bullet(t,\vartheta\setminus t) & \dot{\mathbf{X}}^-_+(t,\vartheta\setminus t) \\
                                                                    0 & \dot{\mathbf{X}}^\bullet_\bullet (t,\vartheta\setminus t)& \dot{\mathbf{X}}^\bullet_+ (t,\vartheta\setminus t) \\
                                                                    0 & 0 & \mathbf{X}(\vartheta\setminus t) \\
                                                                  \end{array}
                                                                \right]=\mathbf{X}(\vartheta)
 \end{equation}
for all $t\in\vartheta$, recovering the conditions for the noiseless case in the absence of the $\bullet$ degree of freedom.

Quantum stochastic calculus is represented in the Hilbert space $\mathcal{H}=\mathfrak{h}\otimes\mathcal{F}$ where $\mathcal{F}=\Gamma(\mathcal{K})$ is a the Guichardet-Fock space over $\mathcal{K}=L^2(\mathbb{R})\otimes\mathfrak{k}$ (see appendix). The Hilbert space $\mathfrak{h}$ is that of an arbitrary quantum system whilst $\mathcal{F}$ serves as a reservoir for the noises that generate the dynamics of the quantum system, and it is the point-derivative  operator that fulfills the analytical requirement of singling out a single point of noise for the purposes of interaction with the system in $\mathfrak{h}$. Quantum stochastic differential equations describe the infinitesimal change on the Hilbert space $\mathcal{H}$ and transitions of the total system in $\mathcal{H}$ are given by operators $\mathrm{T}$ whose action is defined as
\begin{equation}
\big[\mathrm{T}\psi\big](\vartheta^\bullet)=\sum_{\vartheta^\bullet_+\sqcup\vartheta^\bullet_\bullet=\vartheta^\bullet}\int T(\vartheta^\bullet_+,\vartheta^\bullet_\bullet,{\vartheta}^-_\bullet)\psi(\vartheta_\bullet)\mathrm{d}\vartheta^-_\bullet,
\end{equation}
where $\vartheta_\bullet=\vartheta^\bullet_\bullet\sqcup\vartheta^-_\bullet$ and
 the operator-valued kernels $T(\vartheta^\bullet_+,\vartheta^\bullet_\bullet,{\vartheta}^-_\bullet)$ are maps from $\mathfrak{h}\otimes\mathcal{K}^\otimes(\vartheta_\bullet)$ into $\mathfrak{h}\otimes\mathcal{K}^\otimes(\vartheta^\bullet)$. This action is not quite as daunting as it appears for it has three basic parts. The first is the row-like integral action of $\mathrm{T}$ on $\psi$ which is described by the integration over the chain-variable $\vartheta^-_\bullet$, the second is the spontaneous scattering which is given by the purely block-diagonal part of $\mathrm{T}$ whose action has the form $[\mathrm{T}\psi](\vartheta)=\mathrm{T}(\vartheta)\psi(\vartheta)$, and the final part of the action is the creation of additional column-like structure which is accounted for by the chain-variable $\vartheta^\bullet_+$. The algebra of such operators $\mathrm{T}$ on $\mathcal{H}$ is general indeed, and such operators $\mathrm{T}$ have unique block-diagonal representations in $\mathbb{M}$ given by the identification
 \begin{equation}
 \mathbf{T}(\vartheta)^{\iota_1\ldots\;\iota_n}_{\kappa_1\ldots\kappa_n}= \overline{T}({t_1}^{\iota_1}_{\kappa_1},\ldots, {t_n}^{\iota_n}_{\kappa_n})
 \end{equation}
 on the extension $\overline{T}(\vartheta)=1(\vartheta^-_-)T(\vartheta^\bullet_+,\vartheta^\bullet_\bullet,{\vartheta}^-_\bullet,\vartheta^-_+) 1(\vartheta^+_+)\delta_\emptyset(\vartheta^+_-)$ with $t^\iota_\kappa\in\vartheta^\iota_\kappa$. We also define $T(\vartheta^\bullet_+,\vartheta^\bullet_\bullet,{\vartheta}^-_\bullet)=\int T(\vartheta^\bullet_+,\vartheta^\bullet_\bullet,{\vartheta}^-_\bullet,\vartheta^-_+)\mathrm{d}\vartheta^-_+$.
 For more details on norms on these algebras one should consult \cite{Be92b}, it involves working with a rigged Hilbert space constructed from the inductive limits of scaled Hilbert spaces \cite{Be92b,Be91}.

There is a well-defined representation of the $\star$-monoid $\mathbb{M}$ of bounded operators $\mathbf{X}$ on the affine space $\mathbb{M}\boldsymbol{F}^\star\mathcal{H}\subset\mathbb{H}$ in the algebra of operators on $\mathcal{H}$ given by the operator $\boldsymbol{F}^\star$ introduced above. It is $\epsilon(\mathbf{X})=\boldsymbol{F}\mathbf{X}\boldsymbol{F}^\star\equiv\mathrm{X}$ and satisfies the properties
\begin{equation}
\epsilon(\mathbf{X}^\star\mathbf{Y})=\epsilon(\mathbf{X})^\ast\epsilon(\mathbf{Y}), \quad\epsilon(\mathbf{I})=\mathrm{I},\label{morph}
\end{equation}
which is another important theorem given in \cite{Be92b}.
The difference between the operator $\boldsymbol{F}^\star$ used here and that introduced above is that here we are embedding the bigger space $\mathcal{H}$ into the Minkowski-Hilbert space $\mathbb{H}=\mathfrak{h}\otimes\Gamma(\mathbb{K})\cong \mathcal{H}\otimes\Gamma(\Bbbk)$, otherwise the action is still the same: $\boldsymbol{F}^\star\psi=\psi\otimes\xi^\otimes$, and the projection $\boldsymbol{F}$ still corresponds to the integral action of the row-vector ${\xi^\star}^\otimes$. What this all means is that rather than working with complicated operators in $\mathcal{H}$ we can simply embed $\mathcal{H}$ into $\mathbb{H}$ describe dynamical transitions using block-diagonal operators $\mathbf{T}$ and then project the result back into $\mathcal{H}$.

\subsection{Homogeneous QSDEs}
The general form of a quantum stochastic differential equation, or QSDE, is given on the Hilbert space $\mathcal{H}=\mathfrak{h}\otimes\mathcal{F}$ as
\begin{equation}
\mathrm{dT}_{t}=\Lambda\big(\mathbf{D},\mathrm{d}t\big)\equiv\mathrm{D}^\iota_\kappa(t)\mathrm{d}\Lambda^\kappa_\iota \label{abc}
\end{equation}
with some initial condition $\mathrm{T}_0$, and it is an element of the quantum It\^o algebra $\mathfrak{a}(\mathfrak{k})\otimes\mathcal{B}(\mathcal{H})$ having the matrix representation $\mathbf{D}(t)=\pi(\mathrm{dT}_t)$.  This QSDE corresponds to the quantum stochastic single-integral
\begin{equation}
\mathrm{T}_{t}-\mathrm{T}_0 =\int^t_0 \Lambda\big(\mathbf{D},\mathrm{d}z\big) \equiv \boldsymbol{\xi}^\star_t\mathbf{D}\boldsymbol{\xi}_t
\end{equation}
where the right-hand side was introduced in \cite{Brth} for $\boldsymbol{\xi}^\star=(1,\nabla^\ast,0)$ (not to be confused with the temporal-spin $\xi^\star$ which is not in bold-case), with $1\in L^\infty(\mathbb{R})$, to make concise the explicit form of the QS single integral essentially given in \cite{Be91} as the evaluation of
\begin{equation}
\int^t_0 \Lambda\big(\mathbf{D},\mathrm{d}z\big)=\nabla^\ast_t\big( \mathrm{D}^\bullet_\bullet\nabla_t+\mathrm{D}^\bullet_+\big)+ \int_0^t\big[\mathrm{D}^-_\bullet\nabla+\mathrm{D}^-_+\big](z)\mathrm{d}z\label{QSI}
\end{equation}
where $\nabla^\ast$ is defined as $[\nabla^\ast\zeta](\vartheta)=\sum_{t\in\vartheta}\zeta(t,\vartheta\setminus t)$ for $\zeta\in\mathcal{K}\otimes\mathcal{H}$, generalizing the first-order Skorokhod integral \cite{Sko75}, and it is adjoint to the
{point derivative} operator $\nabla:\mathcal{H}\mapsto\mathcal{K}\otimes\mathcal{H}$, similarly defined to that in (\ref{pd}) only this time on $\mathcal{H}$ rather than $\mathbb{M}$, as $[\nabla\chi](t,\vartheta):=\chi(t\sqcup\vartheta)$ for $\chi\in\mathcal{H}$. In fact the operators $\nabla$ and $\nabla^\ast$ are continuous operators in the general setting of Fock-scale \cite{Be92b,Be91}. Also notice that $=\nabla^\ast\nabla=\boldsymbol{\xi}^\star\boldsymbol{\xi}$ is the number operator on $\mathcal{F}$.

We would like to be able to decompose such general differential increments as linear combinations of more basic, or fundamental, processes. Those which we have in mind are
  \emph{homogeneous QSDEs}, having the form
\begin{equation}
\mathrm{dT}_{t}=\Lambda\big(\mathbf{L},\mathrm{d}t\big)\mathrm{T}_t\label{hom}
\end{equation}
corresponding to the  \emph{logarithmic derivative} $\mathbf{D}(x)=\mathbf{L}(x)\mathrm{T}_x$, where $\mathbf{L}$ is called the \emph{chronological exponent} of $\mathrm{T}_t$ \cite{Be92b}. If $\mathrm{T}_t$ resolves  (\ref{hom}) then the block-diagonal operator $\mathbf{T}_t$ from which it is obtained has the \emph{chronologic decomposition}
\begin{equation}
\mathbf{T}_t(\vartheta)=\overset{\leftarrow}{\underset{x\in\vartheta\cap[0,t)}{\prod}} \mathbf{G}(x,\vartheta\setminus x)\mathbf{T}_0(\vartheta),
\end{equation}
and the underlying differential equation describing the dynamics in the Minkowski-Hilbert space is the purely counting differential eqiuation
\begin{equation}
\mathrm{d}\mathbf{T}_t(\vartheta)=\mathbf{L}(t,\vartheta\setminus t)\mathbf{T}_t(\vartheta)\mathrm{d}n_t(\vartheta),\label{genun}
\end{equation}
where $\mathbf{L}(t,\vartheta\setminus t)=\mathbf{G}(t,\vartheta\setminus t)-\mathbf{I}(\vartheta)$. Notice that any invertible operator $\mathbf{T}_t$ defines a general homogeneous process having the chronological exponent $\mathbf{L}(t,\vartheta\setminus t)=\big(\mathrm{d}{\mathbf{T}}_t(\vartheta)\big)\mathbf{T}_t^{-1}(\vartheta),\; t\in\vartheta$. Indeed (\ref{abc}) is also represented in Minkowski-Hilbert space the underlying counting differential equation
\begin{equation}
\mathrm{d}\mathbf{T}_t(\vartheta)=\mathbf{D}(t,\vartheta\setminus t)\mathrm{d}n_t(\vartheta).
\end{equation}
In the case when (\ref{genun}) has the initial condition $\mathbf{T}_0=\mathbf{I}$ the homogeneous process $\mathbf{T}_t$ forms a shift co-cycle $\mathbf{T}_t=\mathbf{T}^t_0$. These are defined as two-parameter families  \begin{equation}
\mathbf{T}^t_r(\vartheta)= \overset{\leftarrow}{\underset{x\in\vartheta\cap[r,t)}{\prod}} \mathbf{G}(x,\vartheta\setminus x), \qquad \mathbf{T}^r_r:=\mathbf{I},\label{hemig}
\end{equation}
satisfying the property $\mathbf{T}^t_r=\mathbf{T}^t_s\mathbf{T}^s_r$ for $r<s<t$.
We may give the more general homogeneous processes $\mathbf{T}_t$, with $\mathbf{T}_0\neq\mathbf{I}$ in terms of these shift co-cycles as $\mathbf{T}_t=\mathbf{T}^t_s\mathbf{T}_s$, for all $t>s$. The projection of the operators $\mathbf{T}^s_r$ onto the Hilbert space $\mathcal{H}$ gives us the two-parameter family $\{\mathrm{T}_r^s=\boldsymbol{F}\mathbf{T}^s_r\boldsymbol{F}^\star\}$.

Notice that the operators $\mathbf{T}^s_r$, and hence the operators $\mathrm{T}_r^s$, need not have their action restricted to the interval $[r,s)$, $r<s$. However, if we insist that $\mathbf{T}^s_r(\vartheta)$ is the identity on $\vartheta\setminus[r,t)$ then we arrive at an import class of co-cycles given by generators of the form
 $\mathbf{G}(x,\vartheta)=\mathbf{G}(x)\otimes\mathbf{I}(\vartheta)$. The processes generated by such are adapted semi-tensor products $\mathbf{T}_t=\mathbf{G}^{\odot}_t$ resolving in $\mathbb{H}$ the underlying QSDE
\begin{equation}
\mathrm{d}\mathbf{T}_{t}(\vartheta)=\mathbf{L}(t)\odot\mathbf{T}_{t}(\vartheta\setminus t)\;\mathrm{d}n_t(\vartheta),\label{318}
\end{equation}
which is what we encountered in the case of Schr\"odinger dynamics in the absence of the $\bullet$ degree of freedom such that $\psi_t=\mathbf{T}_t\psi_0$ resolves (\ref{15}) if $\mathbf{G}=\mathbf{I}-\mathrm{i}H\otimes\pi(\mathrm{d}t)$. We shall also see that a slightly more general choice of $\mathbf{G}$ will give us a unitary dilation of Lindblad dynamics in $\mathbb{H}$.
In fact the most general form of an underlying semi-tensor product propagator $\mathbf{T}_t(\vartheta)=\odot^{\leftarrow}_{x\in\vartheta^t}\mathbf{G}(x)$, $\vartheta^t=[0,t)$,  has a generator of the form
\begin{equation}
\mathbf{G}=\left[
             \begin{array}{ccc}
               I & \mathrm{G}^-_\bullet & G^-_+ \\
               0 & \mathrm{G}^\bullet_\bullet & \mathrm{G}^\bullet_+ \\
               0 & 0 & I \\
             \end{array}
           \right],
\end{equation}
\cite{Be88,Be92,Be92b}, generalizing (\ref{xi}). The requirement for the pseudo-unitarity $\mathbf{G}^{-1}=\mathbf{G}^\star$, where $\mathbf{G}^\star:=\boldsymbol{\eta}^{-1}\mathbf{G}^\ast\boldsymbol{\eta}$ is induced by $\boldsymbol{\eta}$, gives rise to the conditions
\begin{equation}
  {\mathrm{G}_{\bullet}^{\bullet}}^{\ast}\mathrm{G}_{\bullet}^{\bullet
}=\mathrm{I},\qquad{{G}%
_{+}^{-}}^{\dag}+{\mathrm{G}_{+}^{\bullet}}^{\ast}\mathrm{G}_{+}%
^{\bullet}+{G}_{+}^{-}=0,\qquad
{\mathrm{G}_{\bullet}^{-}}^{\ast}+{\mathrm{G}_{\bullet}^{\bullet}}^{\ast
}\mathrm{G}_{+}^{\bullet}=0,\label{three7}
\end{equation}
\cite{Be92b},
allowing us to write such a pseudo-unitary generator as
\begin{equation}
\mathbf{G}=\left[
             \begin{array}{ccc}
               I & -\mathrm{L}^\ast\mathrm{G} & K \\
               0 & \mathrm{G} & \mathrm{L} \\
               0 & 0 & I \\
             \end{array}
           \right],\qquad K=-\mathrm{i}H-\tfrac{1}{2}\mathrm{L}^\ast\mathrm{L}.\label{gun}
\end{equation}
\begin{remark}
 Before we proceed it is worthy of note that an operator $\mathrm{T}=\boldsymbol{F}\mathbf{T}\boldsymbol{F}^\star$  is unitary if and only if the underlying operator $\mathbf{T}$ is pseudo-unitary \cite{Be92b}.
\end{remark}

Now consider a quantum system that is in a pure state given by $\psi\in\mathfrak{h}$, evolving in a homogeneous stochastic manner given by a unitary propagator $\mathrm{U}_t=\boldsymbol{F}\mathbf{G}^\odot_t\boldsymbol{F}^\star$ in $\mathcal{H}$, where $\mathbf{G}$ is a $\star$-unitary generator. We shall also suppose that the initial state-vector in the Guichardet-Fock space $\mathcal{F}$ is the vacuum vector $\delta_\emptyset=0^\otimes$  defined as
\begin{equation}
\delta_\emptyset(\vartheta)=1\;\;\textrm{if}\;\vartheta=\emptyset,\;\;\textrm{otherwise}\;\delta_\emptyset(\vartheta)=0,
\end{equation}
and notice that it is already normalized, $\|\delta_\emptyset\|=\exp\{0\}=1$. Then the underlying dynamics is given by the differential equation (\ref{318}) on an initial state-vector $\Psi_0=\boldsymbol{F}^\star\delta_\emptyset\psi$ in $\mathbb{H}$ as
\begin{equation}
\mathrm{d}\Psi_t(\vartheta)=\mathbf{L}\Psi_t(\vartheta)\mathrm{d}n_t(\vartheta),\label{2006}
\end{equation}
where $\mathbf{L}=\mathbf{G}-\mathbf{I}$. The only difference between this and (\ref{15}) is that
we have the additional $\bullet$ degrees of freedom of the noise, so that the interaction increment is no longer given by $\mathbf{L}=-\mathrm{i}H\otimes\pi(\mathrm{d}t)$, as it was in the case of Schr\"odinger dynamics, but is instead induced by the generator (\ref{gun}),  and the canonical dilation
of this quantum stochastic differential increment has the form
\begin{equation}
\mathbf{L}:\psi\otimes\left[
                        \begin{array}{c}
                          0 \\
                          0 \\
                          1 \\
                        \end{array}
                      \right]\mapsto  \left[
                        \begin{array}{c}
                          K \\
                          \mathrm{L} \\
                          0 \\
                        \end{array}
                      \right]\psi
\end{equation}
generalizing (\ref{17}).

\subsection{Projecting Out Lindblad Dynamics} In the following theorem we show that the marginal density operator $\varrho(t)$ obtained from the density operator $\Psi^{}_t\Psi_t^\star$ by tracing out the Minkowski-Fock space $\mathbb{F}=\Gamma(\mathbb{K})$, where $\Psi_t$ resolves (\ref{2006}), satisfies the Lindblad equation (\ref{lindblad2}).
\begin{proposition}
 Let the $\star$-unitary co-cycle $\mathbf{U}_t:=\mathbf{G}^{\odot}_t$ on $\mathbb{H}$, having initial condition $\mathbf{U}_0=\mathbf{I}$, be given by the $\star$-unitary generator (\ref{gun}) with $\mathrm{G}=\mathrm{I}$ and
the vector-operator $\mathrm{L}$ bounded on any connected $\Delta\subset\mathbb{R}$ as
\begin{equation}
\|\mathrm{L}\|^2=\int_\Delta\|\mathrm{L}(t)^\ast\mathrm{L}(t)\|_{\mathcal{B}(\mathfrak{h}) }^2\mathrm{d}t <\infty,
\end{equation}
with respect to the Lebesgue measure $\mathrm{d}t$ on $\Delta$, and an integrable time-dependent Hamiltonian $H(t)$.
Then the marginal density operator
\emph{\begin{equation}\varrho(t)=\textrm{Tr}_\mathbb{F}[\mathbf{U}_t \big(P_\psi\otimes\Phi\Phi^\star\big)\mathbf{U}_t^\star],\end{equation}}where $P_\psi=\psi\psi^\dag$ and $\Phi:=\boldsymbol{F}^\star\delta_\emptyset\equiv\xi^\otimes$ with $\xi^\star=(1,0,0)$, resolves the Lindblad equation
\begin{equation}
\partial_t\varrho(t)= \varrho(t)K(t)^\dag+\widetilde{\mathrm{L}}(t)\big(\varrho(t)\otimes\mathrm{I}\big) \widetilde{\mathrm{L}}(t)^\ast + K(t)\varrho(t)
\end{equation}
in $\mathfrak{h}$,  where $t\in\Delta$ and $\widetilde{\mathrm{L}}$ is a partial involution of matrix transposition in the algebra $\mathcal{B}(\mathfrak{k})$, but not transposing in $\mathcal{B}(\mathfrak{h})$.
\end{proposition}
\begin{proof}
In the Hilbert space $\mathcal{H}=\mathfrak{h}\otimes\mathcal{F}$ of quantum stochastic evolution we are considering a quantum system embedded in an empty environment so that the initial state-vector is $\psi\otimes\delta_\emptyset$.
On $\mathcal{H}$ we have a geometric quantum Brownian motion given by the propagator $\mathrm{U}_t=\boldsymbol{F}\mathbf{U}_t\boldsymbol{F}^\star$ resolving the quantum stochastic diffusion equation
\begin{equation}
\mathrm{d}\psi_t=K(t)\psi_t\mathrm{d}t+\mathrm{dA}^\star\mathrm{L}(t)\psi_t
\end{equation}
where $\psi_t=\mathrm{U}_t\psi_0$ and we have taken into account that $\mathrm{L}(t)^\ast\mathrm{dA}\psi_t=0$ due to $\psi_0=\psi\otimes\delta_\emptyset$.

The marginal density operator $\varrho(t)=\textrm{Tr}_{\mathbb{F}} \big[\Psi^{}_t\Psi_t^\star\big]$ in $\mathfrak{h}$ may be given in the form
\begin{equation}
\textrm{Tr}_{\mathbb{F}} \big[\mathbf{U}_t(P_\psi\otimes\Phi\Phi^\star)\mathbf{U}_t^\star\big]=\Phi^\star\widetilde{\mathbf{U}}_t \big(P_\psi\otimes{\mathbf{I}}^\otimes\big)\widetilde{\mathbf{U}}_t^\star\Phi,
\end{equation}
where $\Phi=\boldsymbol{F}^\star\delta_\emptyset$, and $\widetilde{\mathbf{U}}_t=\widetilde{\mathbf{G}}_t^\odot$ is a partially transposed propagator given by the generator $\widetilde{\mathbf{G}}:=\boldsymbol{\eta}\mathbf{G}^{\star\dag}\boldsymbol{\eta}$ which is a matrix transposition in $\mathcal{B}(\mathfrak{k})$ of the $\bullet$ degree of freedom, but it {does not} transpose in $\mathcal{B}(\mathfrak{h})$, {nor} the external upper-triangular structure in $\pi\big(\mathfrak{a}_{\boldsymbol{1}}(\mathfrak{k})\big)$, so that the direction of the chronological action is preserved. The generator  $\widetilde{\mathbf{G}}$  has the explicit form
\begin{equation}
\widetilde{\mathbf{G}}=\left[
             \begin{array}{ccc}
               I & \widetilde{\mathrm{L}} & K \\
               0 & {\mathrm{I}} & -\widetilde{\mathrm{L}}^\ast \\
               0 & 0 & I \\
             \end{array}
           \right],
\end{equation}
which is generally not  $\star$-unitary; but one may note that %if $\mathfrak{k}$ is separable then 
$\star$-unitarity of the generator is preserved if the components of the vector operator $\mathrm{L}(t)\in\mathcal{B}(\mathfrak{h})\otimes\mathfrak{k}$ are normal. Now we can define the row-operator ${\mathbf{V}^\odot_t}^\star:=\Phi^\star\widetilde{\mathbf{U}}_t$ where $\mathbf{V}^\star=(I,\widetilde{\mathrm{L}},K)$, from which we find that 
\begin{equation}
\varrho(t)=\varrho(0)+\int_0^t\mathbf{V}^\star(s)\varrho(s)\mathbf{V}(s)\mathrm{d}s.
\end{equation}
The operator $\varrho(t)$ is the  conditional expectation of the operator $\widetilde{\mathbf{U}}_t \big(P_\psi\otimes{\mathbf{I}}^\otimes\big)\widetilde{\mathbf{U}}_t^\star$ with respect to the second-quantization $\Phi$ of the temporal spin $\xi$ (\ref{xinoise}) and its differential increment $\mathrm{d}\varrho(t)$ is %therefore 
the increment of the expected dynamics in $\mathfrak{h}$.
 %which is distinct from the more general dynamics
%\begin{equation}
%\mathrm{d}\rho(t)=\mathbf{V}(t)^\star\rho(t)\mathbf{V}(t)\mathrm{d}t+\Big( %\mathrm{dA}^\star\mathrm{L}\rho(t)+\rho(t)\mathrm{L}^\ast\mathrm{dA}\Big)
%\end{equation}
%which is given more rigorously as the conditional  expectation over $[0,t)\subset\Delta$ as ${\Phi^t}^\star\widetilde{\mathbf{U}}_t \big(P_\psi\otimes{\mathbf{I}}^\otimes\big)\widetilde{\mathbf{U}}_t^\star\Phi^t=\rho(t)$ where $\Phi^t$ is the restriction of $\Phi$ to $\mathbb{F}^t\subset\mathbb{F}$, and the expectation of this stochastic increment coincides with deterministic increment $\mathrm{d}\varrho(t)$ as $l\big(\mathrm{d}\varrho(t)\big)=l\big(\mathrm{d}\rho(t)\big)$.
\end{proof}
In fact, since the Fock-representations of the standard Wiener increments are $\mathrm{d}w^t_k=\mathrm{dA}^k+\mathrm{dA}_k^\star$, $k\in\bullet$, and since $\mathrm{L}(t)^\ast\mathrm{dA}\psi_t=0$ we find that $\psi_t$ also resolves the non-unitary diffusion equation
\begin{equation}
\mathrm{d}\psi_t=K(t)\psi_t\mathrm{d}t+L^k\mathrm{d}w^t_k\psi_t\label{238}
\end{equation}
\cite{Be02}.
In view of (\ref{gun}) this is because $\mathbf{G}\xi$ remains unchanged if $\mathrm{G}^\bullet_\bullet$ and $\mathrm{G}^-_\bullet$ are changed, so we may replace $\mathrm{G}^-_\bullet=-\mathrm{L}^\ast$ with $\widetilde{\mathrm{L}}$ without altering the outcome of the action of $\mathbf{G}$ on $\xi$.

\subsection{Inhomogeneous QSDEs and A Quantum Stochastic Duhamel Principle}
Duhamel's principle is used to solve inhomogeneous differential equations of the form
\begin{equation}
\mathrm{d}{T}(t)+KT(t)\mathrm{d}t=D(t)\mathrm{d}t
\end{equation}
corresponding to generalizations of the homogeneous differential equation
\begin{equation}
\mathrm{d}{Y}(t)+KY(t)\mathrm{d}t=0,\qquad Y(0)=I.
\end{equation}
The Duhamel solution has the form
\begin{equation}
T(t)=Y(t)T(0)+\int_0^tY_z(t)D(z)\mathrm{d}z,
\end{equation}
where $Y_z(t)=\exp\{K(t-z)\}$, and $Y_z(t)D(z)$, $z<t$, resolves the homogeneous differential equation
\begin{equation}
\mathrm{d}Y_z(t)D(z)+KY_z(t)D(z)\mathrm{d}t=0,
\end{equation}
having the $z$-dependent initial condition $Y_z(z)D(z)=D(z)$.

What we would like to do now is consider a quantum stochastic generalization of this to resolve inhomogeneous QSDEs of the form
\begin{equation}
\mathrm{d}\mathrm{T}_t+\Lambda\big(\mathbf{K},\mathrm{d}t\big) \mathrm{T}_t=\Lambda\big(\mathbf{D},\mathrm{d}t\big)\label{312}
\end{equation}
as perturbations of the homogeneous QSDE
\begin{equation}
\mathrm{d}\mathrm{Y}^t_0+\Lambda\big(\mathbf{K},\mathrm{d}t\big) \mathrm{Y}^t_0=0,\qquad\mathrm{Y}^0_0=\mathrm{I}.
\end{equation}
In the following theorem we shall establish such a QS Duhamel principle giving us solutions of (\ref{312}) in the form \begin{equation}
\mathrm{T}_t=\mathrm{Y}^t_0\mathrm{T}_0+ \int_0^t\Lambda\big(\dot{\mathbf{Y}}^t_{\cdot}\mathbf{D},\mathrm{d}z\big),
\end{equation}
where $\big[\dot{\mathbf{Y}}^t_{\cdot}\mathbf{D}\big](z):={\mathrm{Y}}^t_{z}\mathbf{D}(z)$,
for the case when $\mathbf{D}(t)= \mathbf{L}(t){\mathrm{T}}_t$.  From this result we may generalize to the case when $\mathrm{T}_t$ is given as a linear combination $c_k\mathrm{T}^k_t$ and the perturbation derivative has the form $\mathbf{D}(t)=c_k\mathbf{L}^k(t)\mathrm{T}^k_t$.

\begin{lemma}
Let $\mathbf{K}$ be the  chronological exponent of a  two-parameter  family of shift co-cycles
$\{\mathbf{Y}^t_s\}$, which thus solves the underlying counting differential equation
\begin{equation}
\mathrm{d}\mathbf{Y}^{t}_0(\vartheta)=\mathbf{K}(t,\vartheta\setminus t)\mathbf{Y}^t_0(\vartheta)\mathrm{d}n_t(\vartheta), \quad\mathbf{Y}^0_0=\mathbf{I},\label{ab1}
\end{equation}
where $\mathbf{K}=\mathbf{S}-{\mathbf{I}}$. Then the solution $\mathbf{T}_t$ of the underlying counting differential equation
\begin{equation}
\mathrm{d}\mathbf{T}_{t}(\vartheta)=\Big(\mathbf{K}(t,\vartheta\setminus t)+\mathbf{L}(t,\vartheta\setminus t)\Big)\mathbf{T}_t(\vartheta)\mathrm{d}n_t(\vartheta),\label{ab2}
\end{equation}
with arbitrary initial condition $\mathbf{T}_0$, may be given  as
\begin{equation}
\mathbf{T}_t(\vartheta)=\mathbf{Y}^t_0(\vartheta)\mathbf{T}_0(\vartheta)+ \sum_{z\in\vartheta^t}\mathbf{Y}^t_z(\vartheta)\mathbf{L}(z,\vartheta\setminus z)\mathbf{T}_z(\vartheta).\label{ab3}
\end{equation}
\end{lemma}
\begin{proof}
The solution of (\ref{ab3}) may be written as
\begin{equation}\label{bb1}
\mathbf{T}_t(\vartheta)=\sum_{\upsilon\subseteq\vartheta^t}\mathbf{L}^t(t_{|\upsilon|},\vartheta\setminus t_{|\upsilon|})\cdots\mathbf{L}^{t_3}(t_{2},\vartheta\setminus t_{2}) \mathbf{L}^{t_2}(t_{1},\vartheta\setminus t_{1})\mathbf{T}^{t_1}_0(\vartheta)
\end{equation}
with respect to the decomposition $\upsilon=t_{|\upsilon|}\sqcup\ldots\sqcup t_1$ of $\upsilon$ into its disjoint points, where
\begin{equation}\label{bb2}
\mathbf{L}^s(r,\vartheta\setminus r)=\mathbf{Y}^s_r(\vartheta)\mathbf{L}(r,\vartheta\setminus r), \quad \mathbf{T}^t_0:=\mathbf{Y}^t_0\mathbf{T}_0.
\end{equation}
The sum-integrand of (\ref{bb1}) is called the multiple-sum kernel of $\mathbf{T}_t$.
 Now notice that the co-cycles $\mathbf{Y}^s_r$ may be decomposed into a multiple-sum over the exponents $\mathbf{K}(t)$ so that
\begin{equation}\label{bb3}
\mathbf{Y}^s_r(\vartheta)=\sum_{\upsilon\subseteq\vartheta^s_r} \mathbf{K}(t_{|\upsilon|},\vartheta\setminus t_{|\upsilon|})\cdots \mathbf{K}(t_{2},\vartheta\setminus t_{2})\mathbf{K}(t_{1},\vartheta\setminus t_{1}),
\end{equation}
where $\vartheta^s_r=\vartheta\cap[r,s)$, $r<s$.  Thus if we substitute the expressions (\ref{bb2}) into (\ref{bb1}), and then expand (\ref{bb1}) with respect to multiple-sums (\ref{bb3}), we obtain an expression of the form $\mathbf{T}_t(\vartheta)=$
\begin{equation*}
\sum_{\upsilon\sqcup\varkappa\subseteq\vartheta^t} \mathbf{K}(x_{|\varkappa|},\vartheta\setminus x_{|\varkappa|})\mathbf{L}(t_{|\upsilon|},\vartheta\setminus t_{|\upsilon|})\cdots\mathbf{L}(t_{2},\vartheta\setminus t_{2})\mathbf{K}(x_{1},\vartheta\setminus x_{1}) \mathbf{L}(t_{1},\vartheta\setminus t_{1})\mathbf{T}_0(\vartheta),
\end{equation*}
where $\upsilon=t_{|\upsilon|}\sqcup\ldots\sqcup t_1$ and $\varkappa=x_{|\varkappa|}\sqcup\ldots\sqcup x_1$. All we do now is apply the Newton binomial formula to turn this expression into
\begin{equation*}
\mathbf{T}_t(\vartheta)=\sum_{\upsilon\subseteq\vartheta^t} \Big(\mathbf{K}(t_{|\upsilon|},\vartheta\setminus t_{|\upsilon|})+\mathbf{L}(t_{|\upsilon|},\vartheta\setminus t_{|\upsilon|})\Big)\cdots\Big(\mathbf{K}(t_{1},\vartheta\setminus t_{1})+ \mathbf{L}(t_{1},\vartheta\setminus t_{1})\Big)\mathbf{T}_0(\vartheta),
\end{equation*}
which is the solution of (\ref{ab2}).
\end{proof}
The following lemma was given in its original form by Belavkin in \cite{Be92b} as the corollary of a theorem.
\begin{lemma} \emph{(Belavkin)}  Let $\mathbf{D}$ be the representation of a QS derivative on the Minkowski Hilbert space $\mathbb{K}\otimes\mathbb{H}$, then
\begin{equation}
\boldsymbol{F}\boldsymbol{\nabla}^\star\mathbf{D}\boldsymbol{\nabla}\boldsymbol{F}^\star= \boldsymbol{\xi}^\star\boldsymbol{F}\mathbf{D}\boldsymbol{F}^\star\boldsymbol{\xi},
\end{equation}
where $\boldsymbol{\nabla}:\mathbb{H}\rightarrow\mathbb{K}\otimes\mathbb{H}$ is a point-derivative operator on the Minkowski-Hilbert space, and $\boldsymbol{\nabla}^\star$ its adjoint, defining the single-sum process
\begin{equation}
\big[\boldsymbol{\nabla}^\star\mathbf{D}\boldsymbol{\nabla}\big](\vartheta)=\sum_{x\in\vartheta}\mathbf{D}(x, \vartheta\setminus x).
\end{equation}
\end{lemma}
\begin{theorem}
Let $\mathrm{Y}^t_0$ resolve the homogeneous quantum stochastic differential equation
\begin{equation}
\mathrm{d}\mathrm{Y}^{t}_0+\Lambda\big(\mathbf{K},\mathrm{d}t\big)\mathrm{Y}^t_0=0, \quad\mathrm{Y}^0_0=\mathrm{I},\label{1}
\end{equation}
having the chronological exponent $-\mathbf{K}$, then the perturbed quantum stochastic differential equation
\begin{equation}
\mathrm{d}\mathrm{T}_{t}+\Lambda\big(\mathbf{K},\mathrm{d}t\big)\mathrm{T}_t= \Lambda\big(\mathbf{J},\mathrm{d}t\big)\mathrm{T}_t,\label{2}
\end{equation}
has a  quantum stochastic Duhamel solution given as
\begin{equation}
\mathrm{T}_t=\mathrm{Y}^t_0\mathrm{T}_0+ \int^t_0\Lambda\big(\dot{\mathbf{Y}}^t_\cdot\mathbf{J}\dot{\mathbf{T}}_{\cdot},\mathrm{d}z\big),\label{3}
\end{equation}
where $\big[\dot{\mathbf{Y}}^t_\cdot\mathbf{J}\dot{\mathbf{T}}_{\cdot}\big](z): ={\mathrm{Y}}^t_z\mathbf{J}(z){\mathrm{T}}_{z}$.
\end{theorem}
\begin{proof}
The operators $\mathrm{Y}_s^t$ are given as the projections $\boldsymbol{F}\mathbf{Y}_s^t\boldsymbol{F}^\star$ of the underlying block-diagonal operators $\mathbf{Y}_s^t(\vartheta)=\overleftarrow{\prod}_{x\in\vartheta\cap[s,t)}\mathbf{S}(x,\vartheta\setminus x)$ and similarly $\mathrm{T}_t=\boldsymbol{F}\mathbf{G}^\odot_t\mathbf{T}_0\boldsymbol{F}^\star= \boldsymbol{F}\mathbf{G}^\odot_t\boldsymbol{F}^\star\mathrm{T}_0$, where $\mathbf{G}=\mathbf{I}+\mathbf{J}-\mathbf{K}$, and may also be obtained as the projection of the underlying operator $\mathbf{T}_t$ given by (\ref{ab3}) as
\begin{equation}
\mathrm{T}_t=\boldsymbol{F}\mathbf{Y}^t_0\mathbf{T}_0 \boldsymbol{F}^\star+ \boldsymbol{F}\boldsymbol{\nabla}_t^\star \mathbf{Y}^t_\cdot\mathbf{L}\mathbf{T}_\cdot\boldsymbol{\nabla}_t\boldsymbol{F}^\star,
\end{equation}
where $\boldsymbol{F}\mathbf{Y}^t_0\mathbf{T}_0 \boldsymbol{F}^\star=\mathrm{Y}^t_0\mathrm{T}_0$ and $\big[\mathbf{Y}^t_\cdot\mathbf{L}\mathbf{T}_\cdot\big](z,\vartheta\setminus z)= \mathbf{Y}^t_z(\vartheta)\mathbf{L}(z,\vartheta\setminus z)\mathbf{T}_z(\vartheta)$, and since we may write $\boldsymbol{\nabla}\boldsymbol{F}^\star=\boldsymbol{F}^\star\boldsymbol{\xi}$  we obtain our result on noting that $\boldsymbol{F}\mathbf{Y}^t_z\mathbf{L}(z)\mathbf{T}_z \boldsymbol{F}^\star= {\mathrm{Y}}^t_z\mathbf{J}(z)\mathrm{T}_z$ where $\mathbf{J}:=\boldsymbol{F}\mathbf{L}\boldsymbol{F}^\star$.
\end{proof}

\begin{corollary}
Let $\{\mathrm{Y}^s_r\}$ by given by a semi-tensor product two-parameter family of shift co-cycles $\mathbf{Y}^s_r$ on $\mathbb{H}$ having the chronological exponent $-\mathbf{K}$, and let $\mathrm{T}_{t}=\mathrm{T}^t_0\mathrm{T}_0$ be given by another two parameter family of semi-tensor product shift co-cycles $\{\mathbf{T}^s_r\}$ having the  chronological exponent $\mathbf{J}=\boldsymbol{F}\mathbf{L}\boldsymbol{F}^\star$,
then by the QS Duhamel principle we may write the solution $\mathrm{T}_t=\boldsymbol{F}\mathbf{T}_t\boldsymbol{F}^\star$ of the  QSDE
\begin{equation}
\mathrm{d}\mathrm{T}_t+\Lambda\big(\mathbf{K},\mathrm{d}t\big)\mathrm{T}_t= \Lambda\big(\mathbf{J},\mathrm{d}t\big)\mathrm{T}_t,
\end{equation}
with initial condition $\mathrm{T}_0=\boldsymbol{F}\mathbf{T}_0\boldsymbol{F}^\star$, as the projection
\begin{equation}
\mathrm{T}_t=\mathrm{Y}^t_0\mathrm{T}_0+ \boldsymbol{\xi}_t^\star\big(\dot{\mathbf{Y}}^t_\cdot\mathbf{J}\dot{\mathbf{T}}_{\cdot}\big)\boldsymbol{\xi}_t
\end{equation}
of the underlying single-sum process
\begin{equation}
\mathbf{T}_{t}(\vartheta)=\mathbf{Y}_0^t(\vartheta)\mathbf{T}_0(\vartheta) +\sum_{z\in\vartheta^t}\mathbf{Y}^t_z(\vartheta_z)\odot\mathbf{L}(z)\odot \mathbf{T}_{z}(\vartheta^z).
\end{equation}
where  $\vartheta^z=\vartheta\cap[0,z)$ and $\vartheta_z=\vartheta\cap(z,t)$.
\end{corollary}

\section{Quantum Measurement}
It is all well and good to discuss the nature of a quantum system but it is frequently misunderstood that the manner in which a quantum system should be handled is with specific reference to the apparatus which is handling it. The apparatus, as a mathematical object in it own right, is not to be disregarded from the underlying equations of physics. An extensive review of such matters is presented in \cite{Be02} and some of these basic principles are discussed in the Following section.
\subsection{Basic Principles of Measurement} \cite{Be02}
In quantum physics, the preparation of a quantum state may be given by the decoherence of a pure quantum state $P_\psi=|\psi\rangle\langle\psi|$, $\psi\in\mathfrak{h}$, given by the map
\begin{equation}
P_\psi\mapsto\sum_iE_iP_\psi E_i^\dag,\qquad \sum_iE_i^\dag E_i^{}=I\label{21}
\end{equation}
and this may be understood as the conditioning of the quantum system by a complete set of questions that are \emph{compatible} with the apparatus.  That is of course the apparatus which has been used to prepare the quantum state.
Following the preparation of the quantum state, into a decomposition over the potential measurement results, an actual measurement is taken. For example, an actual measurement $i=k$ is taken, at which point the quantum state is assumed to be in the state
\begin{equation}
P_{E_k\psi}=E_k(\psi)P_\psi E_k(\psi)^\dag,\qquad E_k(\psi):=\frac{E_k}{\|E_k\psi\|},
\end{equation}
which is the statistical inference that arises from the measurement.

As a statistical object the apparatus may also be constructed in the Hilbert space formalism. The apparatus is prepared  in a state $|0\rangle\langle0|$ say, $|0\rangle\in\mathfrak{k}$, and in this example it is considered to have a discrete set of outcomes $i\in\bullet=\{0,1,2,\ldots\}$ corresponding to states $|i\rangle\langle i|$ which are the observables of the apparatus. These observables are considered to be the only possible outcomes of the measurement,
and before the measurement is taken we consider the sum of all potential outcomes of the measurement. Each potential outcome is a  conditioning of the quantum object $E_iP_\psi E_i^\dag$ by a positive contraction $E_i$, such that $E_i^\dag E_i^{}\leq I$, corresponding to a potential measurement result, and each such potential result corresponds to a degree of freedom of the apparatus, a 1-dimensional subspace of $\mathfrak{k}$. This is what is understood from
 (\ref{21}): the preparation of the state of the quantum system \emph{prior} to an actual measurement.

The actual measurement is an observation of the apparatus, which is described by the  observable $|k\rangle\langle k|$ if the outcome  $k$ is observed, then at the instant of the measurement we infer that the quantum system is now in the state $P_{E_k\psi}$. Indeed this is considered to be the state of the quantum system as a result of actually taking a measurement.

It is only left to understand how the apparatus is explicitly described in the form of equations, which will ultimately provide a mechanism that explains the rise of decoherence and the transition from possible outcomes to an actual outcome. To this end we consider a bigger system $\mathfrak{h}\otimes\mathfrak{k}$ consisting of both the quantum system under measurement and the  classical apparatus having observables described by a diagonal operator algebra on the apparatus Hilbert space $\mathfrak{k}$. We initially consider a quantum system in a pure state $P_\psi$ composed with a prepared apparatus state $|0\rangle\langle 0|$. The separable compound state is given with respect to tensor product as $P_\psi\otimes|0\rangle\langle 0|$.
In this bigger system the measurement preparation is given by a unitary  operator $\mathrm{G}$ which generates an entangled object-apparatus pure state $\mathrm{G}\big(P_\psi\otimes|0\rangle\langle 0|\big)\mathrm{G}^\ast$. This should be obvious on noting that $P_\psi=\psi\psi^\dag$ from which we immediately see that the entangled pure state assumes the form $\chi\chi^\ast$ where $\chi=\mathrm{G}\big(\psi\otimes|0\rangle\big)$ in $\mathfrak{h}\otimes\mathfrak{k}$.

The row vector $\langle0|$ may assume the form $(1,0,0,\ldots)$ and we may derive the form of $\mathrm{G}$ as
\begin{equation}
\mathrm{G}=\left[
             \begin{array}{cc}
               (I-\mathrm{F}^\ast\mathrm{F})^{1/2} & \mathrm{F}^\ast\\
               \mathrm{F} & -(\mathrm{I}-\mathrm{F}\mathrm{F}^\ast)^{1/2} \\
             \end{array}
           \right]\label{G}
\end{equation}
where $\mathrm{F}^\ast=\sum_i\langle i|E^\dag_i$ is a row vector of the operators $E^\dag_i$ on the quantum object over $i=1,2,\ldots$, having co-dimension 1 in  the apparatus space $\mathfrak{k}$. Also note that $\mathrm{F}\mathrm{F}^\ast$ is a matrix $\big[E_iE_j^\dag\big]$ of operators, whilst $\mathrm{F}^\ast\mathrm{F}$ is in $\mathcal{B}(\mathfrak{h})$, and $(I-\mathrm{F}^\ast\mathrm{F})^{1/2}=E_0$.

The next thing to notice is that when we trace over the apparatus Hilbert space, to find the marginal density of the quantum system resulting from this entanglement interaction, we see that it is indeed the decoherence of the quantum state $P_\psi$. So what we are saying is that in the total system consisting of both the apparatus and the quantum object we obtain the decoherence of the quantum object as the marginal state of the joint system. That is
\begin{equation}
\textrm{Tr}_{\mathfrak{k}}\big[\mathrm{G}\big(P_\psi\otimes|0\rangle\langle 0|\big)\mathrm{G}^\ast\big]= \sum_iE_iP_\psi E_i^\dag.\end{equation}
As Belavkin so eloquently states
\begin{quote}
\emph{The quantum state decoherence is simply obtained in the ignorance of the apparatus.}
\end{quote}
The procedure of reading the apparatus result is given mathematically by the action of the orthoprojector $|k\rangle\langle k|$ which corresponds to the reality: outcome $k$ of the apparatus is observed. Following this, the entangled pure state is operated on by the projector $I\otimes |k\rangle\langle k|$ giving us
\begin{equation}
\big(I\otimes |k\rangle\langle k|\big)\mathrm{G}\big(P_\psi\otimes|0\rangle\langle 0|\big)\mathrm{G}^\ast \big(I\otimes |k\rangle\langle k|\big)=E_kP_\psi E_k^\dag\otimes |k\rangle\langle k|,
\end{equation}
such that we observe the apparatus reading $k$ and understand that the quantum object has been conditioned by such an observation. Then  re-normalization gives us the result $P_{E_k\psi}\otimes|\mathrm{k}\rangle\langle k|$. This is the non-linear quantum filtering on the level of a single measurement; note that $P_{E_k\psi}$ is a non-linear function of $\psi$ due to the re-normalization.

%\subsection{Non-Demolition Principle}
To finish off this brief introduction to quantum measurement and filtering we shall explain what is required of the apparatus and quantum object in order to make this inference possible. In other words, what is required in order to apply what is essentially Bayes formula.

Consider an arbitrary question, or proposition, about the future of the total system; this may be given by a future-event orthoprojector $\mathrm{P}$. This question must be compatible with the historical data, which in this case is the single actual-event given by an orthoprojector $\mathrm{M}$, the memory. The reason why we insist on such compatibility of potential future events with the actual past events  is because we must insist that the quantum probability of the proposition $\mathrm{P}$ may be given as the weighted sum
\begin{equation}
\textrm{Pr}\big[\mathrm{P}\big]=\textrm{Pr}\big[\mathrm{P}|\mathrm{M}\big]\textrm{Pr}\big[\mathrm{M}\big]+ \textrm{Pr}\big[\mathrm{P}|\mathrm{I}-\mathrm{M}\big]\textrm{Pr}\big[\mathrm{I}-\mathrm{M}\big],\label{22}
\end{equation}
of the conditional probability $\textrm{Pr}\big[\mathrm{P}|\mathrm{M}\big]$ and the complimentary conditional probability $\textrm{Pr}\big[\mathrm{P}|\mathrm{I}-\mathrm{M}\big]$. \emph{This means that we only consider the measurable events $\mathrm{M}$ which serve as conditions for any future propositions} $\mathrm{P}$. If (\ref{22}) is not satisfied then, as Belavkin states, we do not have statistical causality. Ultimately this amounts to Belavkin's \emph{Non-demolition Principle} of quantum measurement, a principle of quantum causality \cite{Be94}.
\begin{remark}
The Bayesian conditional probability formula is
\emph{\begin{equation}
\textrm{Pr}\big[\mathrm{P}|\mathrm{M}\big]=\frac{\textrm{Pr}\big[\inf\{\mathrm{P},\mathrm{M}\}\big]} {\textrm{Pr}\big[\mathrm{M}\big]}\quad \forall\; \mathrm{P},
\end{equation}}
where $\inf\{\mathrm{P},\mathrm{M}\}$ is the largest orthoprojector $\mathrm{E}$ such that both $\mathrm{P}-\mathrm{E}$ and $\mathrm{M}-\mathrm{E}$ are positive operators. The condition (\ref{22}) means that $\inf\{\mathrm{P},\mathrm{M}\}$ is a linear function in the second variable, so that
 \begin{equation}\inf\{\mathrm{P},\mathrm{M}\}+ \inf\{\mathrm{P},\mathrm{I}-\mathrm{M}\}=\inf\{\mathrm{P},\mathrm{I}\}=\mathrm{P},
 \end{equation}
 which is only satisfied if $[\mathrm{P},\mathrm{M}]:=\mathrm{PM}-\mathrm{MP}=0$.
\end{remark}
 \subsection{A Sequential Measurement Process}
 Now we shall consider a sequence of such measurements described above. This begins with an entangling dynamics $\mathrm{U}_t:\psi_0\mapsto\psi_t$ generated by the $\star$-unitary interaction operator
\begin{equation}
\mathbf{G}=\left[
             \begin{array}{ccc}
               I & 0 & -\mathrm{i}H \\
               0 & \mathrm{G} & 0 \\
               0 & 0 & I \\
             \end{array}
           \right]\label{inG}
\end{equation}
such that $\mathrm{U}_t=\boldsymbol{F}\mathbf{G}_t^{\odot}\boldsymbol{F}^\star$, having the explicit block diagonal form
\begin{equation}
 \mathrm{U}_t(\vartheta)=\exp\{-\mathrm{i}Ht\}\underset{{x\in\vartheta^t}}{\overset{\leftarrow}{\odot}} \mathrm{G}(x),
 \end{equation}
  where
  $\mathrm{G}(x):=\exp\{\mathrm{i}Hx\}\mathrm{G}\exp\{-\mathrm{i}Hx\}$, describing a process of spontaneous object-apparatus interactions with a Hamiltonian evolution of the object between interactions. The initial state of the interaction dynamics is taken to be $\psi_0=\psi\otimes\mathrm{f}^\otimes$, where $\mathrm{f}=|0\rangle$, which is normalized with respect to the Poisson measure $\mathrm{d}\texttt{P}^t_\nu(\vartheta)=\nu^\otimes(\vartheta)\mathrm{d}\vartheta\exp\{-\nu t\}$ as
  \begin{equation}
  \|\mathrm{f}^\otimes\|^2_\nu=\int\|\mathrm{f}^\otimes(\vartheta)\|^2\nu^\otimes(\vartheta) \mathrm{d}\texttt{P}^t_\nu(\vartheta)=1.
  \end{equation}
 The Poisson intensity $\nu$ is considered to be a positive scalar and describes the object-apparatus interaction frequency, and this formalism is called the Guichardet-Fock space representation of Poisson space, denoted by $\mathcal{G}_\nu^t$ (see appendix for more details). The QSDE describing this interaction dynamics is a counting differential equation having the form
    \begin{equation}
 \mathrm{d}\psi_t(\vartheta)+\mathrm{i}H\psi_t(\vartheta)\mathrm{d}t= \mathrm{J}\psi_t(\vartheta)\mathrm{d}n_t(\vartheta),\label{49b}
 \end{equation}
 where $\mathrm{J}=\mathrm{G}-\mathrm{I}$.

 The same dynamics may be considered on the Hilbert space $\mathcal{H}^t$ with respect to the initial state $\psi_0=\psi\otimes\mathrm{g}^\otimes\exp\{-\frac{1}{2}\mathrm{g}^\ast\mathrm{g}\}$, where $\mathrm{g}=\sqrt{\nu}\mathrm{f}$. This is in fact a Weyl transformation of the Fock-vacuum: $\mathrm{g}^\otimes\exp\{-\frac{1}{2}\mathrm{g}^\ast\mathrm{g}\}=\mathrm{W}_\mathrm{g}\delta_\emptyset$ (see appendix), which means that we can reconsider the dynamics on the Fock-vacuum if we transform the propagator $\mathrm{U}_t$ into $\mathrm{W}_\mathrm{g}^\ast\mathrm{U}_t\mathrm{W}_g$, having the underlying generator
 \begin{equation}
 \mathbf{Z}^\star_\mathrm{g}\mathbf{G}\mathbf{Z}_\mathrm{g}=\left[
                                                              \begin{array}{ccc}
                                                                I & \mathrm{g}^\ast\mathrm{J} & \mathrm{g}^\ast\mathrm{J}\mathrm{g}-iH \\
                                                                0 & \mathrm{G} & \mathrm{J}\mathrm{g} \\
                                                                0 & 0 & I \\
                                                              \end{array}
                                                            \right],\label{poissonweyl}
 \end{equation}
 where $\mathrm{W}_\mathrm{g}=\boldsymbol{F}\mathbf{Z}_\mathrm{g}^\otimes\boldsymbol{F}^\star$ and $\mathbf{Z}_\mathrm{g}$ is given by (\ref{Z}).
 If we define the component operators $G^k_0=\langle k|\mathrm{G}|0\rangle$, $k\neq0$, as $L^k/\sqrt{\nu}$, where the operators $L^k$ are not depending on $\nu$, then we find that in the central limit $\sqrt{\nu}\nearrow\infty$ the wave-function that was resolving (\ref{49b}) now resolves the quantum state diffusion equation
   \begin{equation}
 \mathrm{d}\psi_t+\big(\tfrac{1}{2}\mathrm{L}^\ast\mathrm{L}+\mathrm{i}H\big)\psi_t\mathrm{d}t= L^k\mathrm{d}w_k^t\psi_t,\quad \psi_0=\psi\otimes\delta_\emptyset\label{49c}
 \end{equation}
 \cite{Be93,BelM96}, and note that $\mathrm{L}=\sum_{k\neq0}L^k|k\rangle$ does not include the initial state of the apparatus so that all output trajectories are transformations into new apparatus states.

 What we would like to do next is return to (\ref{49b}) and consider each object-apparatus interaction as a map from an input space $\mathfrak{h}\otimes\check{\mathfrak{k}}_\nu$ into an output space $\mathfrak{h}\otimes\hat{\mathfrak{k}}_\nu$, where $\check{\mathfrak{k}}_\nu$ is the Hilbert space $\mathfrak{k}$ equipped with the scalar metric $\check{\nu}=\nu\mathrm{I}$ and $\hat{\mathfrak{k}}_\nu$ is the Hilbert space $\mathfrak{k}$ equipped with the diagonal Riemannian metric
 \begin{equation}
\hat{\nu}_\psi(t)=\left[
            \begin{array}{ccc}
              \nu_0(t) & 0 & \cdots \\
              0 & \nu_1(t) & {} \\
              \vdots & {} & \ddots \\
            \end{array}
          \right],\label{curve}
\end{equation}
where $\nu_k(t)$ are the individual intensities of the observable output trajectories coupled to the quantum system.
They are given by re-scalings of the input frequency $\nu$ by the probabilities $\|G^k_0\psi_t\|^2$ of the operations $G^k_0\psi_t$ so that $\nu_k(t)=\nu\|G^k_0\psi_t\|^2$ from which it follows that $\sum_k\nu_k(t)=\nu$.

Let's denote by  $\mathrm{G}(\psi)$ the interaction operator mapping $\mathfrak{h}\otimes\check{\mathfrak{k}}_\nu$ into $\mathfrak{h}\otimes\hat{\mathfrak{k}}_\nu$. It is obtained from $\mathrm{G}$ by the transformation $\mathrm{G}(\psi)=\sqrt{\hat{\nu}}^{-1}\mathrm{G}\sqrt{\check{\nu}}$
and satisfies the unusual unitarity condition
\begin{equation}
\check{\nu}^{-1}\mathrm{G}(\psi)^\ast\hat{\nu}=\mathrm{G}(\psi)^{-1}.
\end{equation}
What we also find when considering the interaction operator $\mathrm{G}(\psi)$ is that its action on the initial apparatus state $\mathrm{f}=|0\rangle$ gives rise to the conditioned operators $G^k_0/\|G^k_0\psi_t\|=\langle k|\mathrm{G}(\psi)|0\rangle$, so that the transformation from $\mathfrak{h}\otimes\check{\mathfrak{k}}_\nu$ to $\mathfrak{h}\otimes\hat{\mathfrak{k}}_\nu$ describes the filtering. In order to consider the QSDE that describes this dynamics one should first notice that, in this case, $\psi_0$ is in $\check{\mathcal{G}}_\nu$ whereas $\hat{\psi}_t=\mathrm{U}_t(\psi)\psi_0$ is in the product space $\hat{\mathcal{G}}_\nu^t\otimes\check{\mathcal{F}}_\nu$; indeed $\psi_t=\psi^t\otimes\mathrm{f}^\otimes$.
We may write down the QSDE for the output dynamics as
  \begin{equation}
 \mathrm{d}\hat{\psi}_t(\vartheta)+\mathrm{i}H\hat{\psi}_t(\vartheta)\mathrm{d}t= \mathrm{J}(\psi)\hat{\psi}_t(\vartheta)\mathrm{d}n_t(\vartheta)\label{outeq}
 \end{equation}
 where $\mathrm{J}(\psi)=\mathrm{G}(\psi)-\sqrt{\hat{\nu}}^{-1}\sqrt{\hat{\nu}}$.

 The actual measurements corresponding to observations of the output trajectories of the unitary counting dynamics (\ref{49b}) may be given by the measurement orthoprojector
 \begin{equation}
\mathrm{M}_t(\vartheta)=I\otimes\mathrm{M}^\otimes(\vartheta^t)\otimes\mathrm{I}^\otimes(\vartheta\setminus\vartheta^t),
\end{equation}
where $\mathrm{M}(t)\in\{|k\rangle\langle k|,k=0,1,2,\ldots\}$ and depends solely on what is actually obsserved.
 This defines a new wave-vector $\chi_t:=\mathrm{M}_t\psi_t$ for the actual measurement dynamics, which is not unitary, and it resolves the QSDE
\begin{equation}
\mathrm{d}\chi_t(\vartheta)+\mathrm{i}H\chi_t(\vartheta)\mathrm{d}t= (\mathrm{M}(t)\mathrm{G}-\mathrm{I})\chi_t(\vartheta)\mathrm{d}n_t(\vartheta).\label{Min}
\end{equation}
The action of the observable $\mathrm{N}^t_k=\nabla^\ast_t|k\rangle\langle k|\nabla_t$ on the vector $\chi_t$ gives the number $n^t_k$ of measurements of type $k$ up to the point $t$ as $[\mathrm{N}^t_k\chi_t](\vartheta)=n^t_k(\vartheta)\chi_t(\vartheta)$, whilst the action of $\mathrm{N}^t_k$ on $\psi_t$ does not. Further, notice that the wave-function $\psi_t$ describes an entanglement of the quantum object with the apparatus up to the point $t$, whereas the actual measurements disentangle the object from the apparatus as $\chi_t(\vartheta)$ has the form
\begin{equation}
\chi_t(\vartheta)=|k_1\rangle\otimes |k_2\rangle\cdots\otimes|k_n\rangle\otimes \chi(t)\otimes |0\rangle \cdots\otimes|0\rangle
\end{equation}
where $\vartheta^t=\{t_1<\ldots<t_n\}$ and $\chi(t)$ is the projection $\langle\mathrm{M}_t\varphi|\chi_t\rangle_\nu$ of $\chi(t)$ into $\mathfrak{h}$, where $\varphi^\ast=(1,1,1,\ldots)^\otimes$.

Finally notice that in the output space the measurement dynamics assumes the form
\begin{equation}
\mathrm{d}\hat{\chi}_t(\vartheta)+\mathrm{i}H\hat{\chi}_t(\vartheta)\mathrm{d}t= \Big(\mathrm{M}(t)\mathrm{G}/\| {G}^{k(t)}_0\psi(t)\| -\sqrt{\tfrac{\nu}{\nu_k(t)}}\mathrm{I}\Big)\hat{\chi}_t(\vartheta)\mathrm{d}n_t(\vartheta),
\end{equation}
where $\hat{\chi}_t(\vartheta)= |k_1\rangle\otimes |k_2\rangle\cdots\otimes|k_n\rangle\otimes \psi(t)\otimes |0\rangle \cdots\otimes|0\rangle$ for the normalized wave-function $\psi(t)\in\mathfrak{h}$ having the explicit form
\begin{equation}
\psi(t)=\exp\{-\mathrm{i}Ht\}G^{k_n}_0(t_n)\cdots G^{k_2}_0(t_2)G^{k_1}_0(t_1)\psi,
\end{equation}
where $G^{k_i}_0(t_i)=\exp\{\mathrm{i}Ht_i\}\big(G^{k_i}_0/\|G^{k_i}_0\psi(t_i)\|\big)\exp\{-\mathrm{i}Ht_i\}$.

\subsection{Pseudo-Measurement}
 Now we shall consider the idea of a quantum system being prepared for evolution by a pseudo-apparatus.  In the simplest case, corresponding to the dilation of Schr\"odinger dynamics, this pseudo-apparatus is represented by the Minkowski-Hilbert space $\Bbbk(t)=(\mathbb{C}^2,\boldsymbol{\eta})$ and the preparation of the quantum state $P_\psi$ is now given by a pseudo-decoherence of the quantum system which is just the map
\begin{equation}
P_\psi\mapsto\textrm{Tr}_\Bbbk\big[\mathbf{G}\big(P_\psi\otimes\xi\xi^\star\big)\mathbf{G}^\star\big]\equiv\dot{P}_\psi
\end{equation}
transforming $P_\psi$ to its traceless derivative $\dot{P}_\psi:=\mathrm{i}[P_\psi,H]$.
Although the geometry here is hyperbolic there is no problem with finding a $\star$-orthoprojector $\mathbf{M}=\mathbf{M}^\star\mathbf{M}=\mathbf{M}^\star$ to define an observation in $\Bbbk$ corresponding to a transformation of the quantum state to its derivative in $\mathfrak{h}$. The $\star$-projector that we are looking for is simply the identity operator $\mathbf{M}=\mathbf{I}$ trivially giving us the measurement
\begin{equation}
\mathbf{M}\mathbf{G}\big(P_\psi\otimes\xi\xi^\star\big)\mathbf{G}^\star\mathbf{M}= \mathbf{G}\big(P_\psi\otimes\xi\xi^\star\big)\mathbf{G}^\star
\end{equation}
as oppose to the ortho-complimentary statement $\mathbf{M}=\mathbf{O}$ simply giving us zero.

Indeed a process of such measurements is given by (\ref{15}) but in addition to the dynamics considered in the first chapter we shall now suppose that we have a counting process of intensity $\nu$ so that  the solution (\ref{112}) is no longer regarded as an element of $\mathbb{H}$ but instead as an element of the Minkowski-Poisson space $\mathbb{G}_\nu$ (see appendix) having the Poisson measure $\mathrm{d}\texttt{P}_{2\nu}^t(\vartheta)$ on the subspaces $\mathbb{G}_\nu^t\subset\mathbb{G}_\nu$ for a scalar Poisson intensity that we shall take to be $2\nu$ for convenience. One may have wondered when and where the Lorentz transform might appear in a theory of physics in a Minkowski space, and  in the following  theorem regarding the dynamics (\ref{15}) in Minkowski-Poisson space, we shall see how the Lorentz transform appears as a transformation between Poisson intensities i.e. it is a transformation that changes the frequency  of observation of the quantum system.

\begin{proposition}
Let $\mathbf{X}_t$ be an element of the $\star$-monoid $\mathbb{M}$ generated by the Newton-Leibniz algebra $\mathfrak{a}\otimes\mathcal{B}(\mathfrak{h})$, then the pseudo-Poisson expectation \emph{
\begin{equation}
\texttt{P}_{2\nu}^t\big[\mathbf{X}_t\big]:=\int\varphi^\star(\vartheta)\mathbf{X}_t(\vartheta)\varphi(\vartheta) \mathrm{d}\texttt{P}_{2\nu}^t(\vartheta)
\end{equation}}into $\mathcal{B}(\mathfrak{h})$, where $\mathrm{d}\texttt{P}_{2\nu}^t(\vartheta)=(2\nu)^\otimes(\vartheta)\mathrm{d}\vartheta\exp\{-2\nu t\}$ and $\varphi^\star=(\frac{1}{\sqrt{2}},\frac{1}{\sqrt{2}})^\otimes$, may be given by the Lorentz transformed projection $\boldsymbol{F}_\nu=\boldsymbol{F}{\boldsymbol{\upsilon}_\nu^\star}^\otimes$, where $\boldsymbol{\upsilon}_\nu$ is the real Lorentz transform
\begin{equation}
\boldsymbol{\upsilon}_\nu= \left[
  \begin{array}{cc}
    \frac{1}{\sqrt\nu} & 0 \\
    0 & \sqrt{\nu} \\
  \end{array}
\right],
\end{equation}
as \emph{
\begin{equation}
\texttt{P}_{2\nu}^t\big[\mathbf{X}_t\big]=\boldsymbol{F}_\nu\mathbf{X}_t \boldsymbol{F}_\nu^\star.
\end{equation}}Further, the pseudo-Poisson expectation of the  process $\mathbf{X}_t=\mathbf{G}^{\odot}_t$ resolves the boosted Schr\"odinger equation \emph{
\begin{equation}
\mathrm{d}\psi(t)=-\mathrm{i}H\psi(t)\nu\mathrm{d}t,\quad \psi(t)=\texttt{P}_{2\nu}^t\big[\mathbf{G}^{\odot}_t \big]\psi,\label{boosted}
\end{equation}}where $\pi(\nu\mathrm{d}t)=\upsilon^\star_\nu\pi(\mathrm{d}t)\upsilon_\nu$ is called the boosted time increment.
\end{proposition}
\begin{proof}
First we note that $\sqrt{2\nu}^\otimes\mathbf{f}^\otimes\exp\{-\nu t\}$, for $\mathbf{f}^\otimes=\varphi$, may be obtained from a pseudo-Weyl transform of the vacuum vector $\boldsymbol{\delta}_\emptyset$ in $\mathbb{H}$ as $\mathbf{W}_{\mathbf{g}}\boldsymbol{\delta}_\emptyset$ if $\mathbf{g}=\sqrt{2\nu}\mathbf{f}$. Next we notice that since $\mathbf{W}_{\mathbf{h}+\mathbf{k}}=\mathbf{W}_{\mathbf{h}}\mathbf{W}_{\mathbf{k}}$, and since $\mathbf{W}_{\mathbf{h}}^\star \mathbf{X}\mathbf{W}_{\mathbf{h}}=\mathbf{X}$ for any $\mathbf{h}=(0,h)^\star$ because $\mathbf{X}(t\sqcup\vartheta)\mathbf{h}(t)=\mathbf{h}(t)\mathbf{X}(\vartheta)$ for such $\mathbf{h}$, we have \begin{equation}
\boldsymbol{\delta}_\emptyset^\star\mathbf{W}_{\mathbf{g}}^\star\mathbf{X} \mathbf{W}_{\mathbf{g}}\boldsymbol{\delta}_\emptyset= \boldsymbol{\delta}_\emptyset^\star\mathbf{W}_{\mathbf{g}^+}^\star\mathbf{X} \mathbf{W}_{\mathbf{g}^+}\boldsymbol{\delta}_\emptyset \equiv \boldsymbol{F}_\nu\mathbf{X}\boldsymbol{F}_\nu^\star,
\end{equation}
where $\mathbf{g}^+=(\sqrt{\nu},0)^\star$ and $\boldsymbol{F}^\star_\nu=\boldsymbol{\upsilon}_\nu^\otimes\boldsymbol{F}^\star$. %More details may be found in \cite{Br12}.
To understand that the pseudo-Poisson expectation of $\mathbf{G}^{\odot}_t$ resolves the boosted Schr\"odinger equation one should simply notice that $\boldsymbol{F}_\nu\mathbf{G}^{\odot}_t\boldsymbol{F}_\nu^\star=\boldsymbol{F}\big(\boldsymbol{\upsilon}^\star_\nu \mathbf{G} \boldsymbol{\upsilon}_\nu\big)^{\odot}_t \boldsymbol{F}^\star$ where $\boldsymbol{\upsilon}^\star_\nu \mathbf{G} \boldsymbol{\upsilon}_\nu=\mathbf{I}-\mathrm{i}H \boldsymbol{\upsilon}^\star_\nu \pi(\mathrm{d}t) \boldsymbol{\upsilon}_\nu$, and $\boldsymbol{\upsilon}^\star_\nu \pi(\mathrm{d}t) \boldsymbol{\upsilon}_\nu=\nu\pi(\mathrm{d}t)$ which may be explicitly seen as
\begin{equation}
\left[
  \begin{array}{cc}
    \sqrt{\nu} & 0 \\
    0 & \frac{1}{\sqrt\nu} \\
  \end{array}
\right]\left[
  \begin{array}{cc}
    0 & 1 \\
    0 & 0 \\
  \end{array}
\right]\left[
  \begin{array}{cc}
    \frac{1}{\sqrt\nu} & 0 \\
    0 & \sqrt{\nu} \\
  \end{array}
\right]=\left[
  \begin{array}{cc}
    0 & \nu \\
    0 & 0 \\
  \end{array}
\right].
\end{equation}
\end{proof}
Now that we have introduced a notion of Poisson intensity in the underlying Minkowski-Hilbert space  we shall see that things appear a little differently to the sequential measurement dynamics discussed above.

\begin{proposition}
Let $\mathbf{X}_t$ an element of the $\star$-monoid $\mathbb{M}$ generated by the quantum It\^o algebra $\mathfrak{a}(\mathfrak{k})\otimes\mathcal{B}(\mathfrak{h})$ then the pseudo-Poisson expectation \emph{
\begin{equation}
\texttt{P}_{3\nu}^t\big[\mathbf{X}_t\big]:=\int\varphi^\star(\vartheta)\mathbf{X}_t(\vartheta)\varphi(\vartheta) \mathrm{d}\texttt{P}_{3\nu}^t(\vartheta),
\end{equation}}where $\mathrm{d}\texttt{P}_{3\nu}^t(\vartheta)=(3\nu)^\otimes(\vartheta)\mathrm{d}\vartheta\exp\{-3\nu t\}$ and $\varphi=\mathbf{f}^\otimes$ and $\mathbf{f}^\star=\frac{1}{\sqrt{3}}(1,\mathrm{f}^\ast,1)$, may be given by the Lorentz transformed projection $\boldsymbol{F}_\nu=\boldsymbol{F}{\boldsymbol{\upsilon}_\nu^\star}^\otimes$, where $\boldsymbol{\upsilon}_\nu$ is the real Lorentz transform
\begin{equation}
\boldsymbol{\upsilon}_\nu=\left[
                            \begin{array}{ccc}
                              \tfrac{1}{\sqrt\nu} & 0 & 0 \\
                              0 & \mathrm{I} & 0 \\
                              0 & 0 & \sqrt{\nu} \\
                            \end{array}
                          \right]\label{lorentz3}
\end{equation}
 as \emph{
\begin{equation}
\texttt{P}_{3\nu}^t\big[\mathbf{X}_t\big]=\texttt{E}_\emptyset\big[ \boldsymbol{F}_\nu{\mathbf{Z}_\mathrm{f}^\otimes}^\star\mathbf{X}_t\mathbf{Z}_\mathrm{f}^\otimes \boldsymbol{F}_\nu^\star\big],
\end{equation}}where {$\texttt{E}_\emptyset[\mathrm{X}]=\delta_\emptyset^\ast\mathrm{X}\delta_\emptyset$} and $\boldsymbol{F}_\nu{\mathbf{Z}_\mathrm{f}^\otimes} \boldsymbol{F}_\nu^\star=\mathrm{W}_{\sqrt{\nu}\mathrm{f}}$. Further, the pseudo-Poisson expectation of the underlying process $\mathbf{X}_t={\widetilde{\mathbf{G}}^{\odot}_t(\varrho\otimes\mathbf{I}^\otimes) {\widetilde{\mathbf{G}}_t^{\odot^\star}}}$, with $\widetilde{\mathbf{G}}$ given by (\ref{lindblad}), resolves the boosted Lindblad equation \emph{
\begin{equation}
\mathrm{d}\varrho(t)=\mathbf{V}^\star\varrho(t)\mathbf{V}\nu\mathrm{d}t,\quad \varrho(t)=\texttt{P}_{3\nu}^t\big[\mathbf{X}_t \big].\label{bl}
\end{equation}}
\end{proposition}
\begin{proof}
 Again we note that $\sqrt{3\nu}^\otimes\mathbf{f}^\otimes\exp\{-\frac{3}{2}\nu t\}$ may be obtained from a pseudo-Weyl transform of the vacuum vector $\boldsymbol{\delta}_\emptyset$ in $\mathbb{H}$ as $\mathbf{W}_{\mathbf{g}}\boldsymbol{\delta}_\emptyset$ if $\mathbf{g}=\sqrt{3\nu}\mathbf{f}$, and again we notice that since $\mathbf{W}_{\mathbf{h}+\mathbf{k}}=\mathbf{W}_{\mathbf{h}}\mathbf{W}_{\mathbf{k}}$, and since $\mathbf{W}_{\mathbf{h}}^\star \mathbf{X}\mathbf{W}_{\mathbf{h}}=\mathbf{X}$ for any $\mathbf{h}=(0,0,h)^\star$ because $\mathbf{X}(t\sqcup\vartheta)\mathbf{h}(t)=\mathbf{h}(t)\mathbf{X}(\vartheta)$ for such $\mathbf{h}$, we have \begin{equation}
\boldsymbol{\delta}_\emptyset^\star\mathbf{W}_{\mathbf{g}}^\star\mathbf{X} \mathbf{W}_{\mathbf{g}}\boldsymbol{\delta}_\emptyset=\delta_\emptyset^\ast \boldsymbol{F}_\nu\mathbf{W}_{\mathbf{g}^\bullet}^\star\mathbf{X}\mathbf{W}_{\mathbf{g}^\bullet}\boldsymbol{F}_\nu^\star \delta_\emptyset,
\end{equation}
where $\mathbf{g}^\bullet:=(0,\mathrm{g}^\ast,0)^\star$. Further, with respect to the Lorentz transformed $\star$-orthoprojector $\mathbf{E}_\nu={\boldsymbol{\upsilon}_\nu^\otimes}^\star\boldsymbol{F}^\star \boldsymbol{F}\boldsymbol{\upsilon}_\nu^\otimes$ we may write
\begin{equation}
\delta_\emptyset^\ast \boldsymbol{F}_\nu\mathbf{W}_{\mathbf{g}^\bullet}^\star\mathbf{X}\mathbf{W}_{\mathbf{g}^\bullet}\boldsymbol{F}_\nu^\star \delta_\emptyset=
\delta_\emptyset^\ast \boldsymbol{F}_\nu\mathbf{W}_{\mathbf{g}^\bullet}^\star\mathbf{E}_\nu\mathbf{X} \mathbf{E}_\nu\mathbf{W}_{\mathbf{g}^\bullet}\boldsymbol{F}_\nu^\star \delta_\emptyset,
\end{equation}
which may be verified by expansion.
 Finally, $\mathbf{E}_\nu\mathbf{W}_{\mathbf{g}^\bullet}\boldsymbol{F}_\nu^\star \delta_\emptyset= \boldsymbol{F}_\nu^\star\mathrm{W}_{\mathrm{g}}\delta_\emptyset=\mathbf{E}_\nu \mathbf{Z}_\mathrm{f}^\otimes \boldsymbol{F}_\nu^\star\delta_\emptyset$ so that we may write
\begin{equation}
\delta_\emptyset^\ast \boldsymbol{F}_\nu\mathbf{W}_{\mathbf{g}^\bullet}^\star\mathbf{E}_\nu\mathbf{X} \mathbf{E}_\nu\mathbf{W}_{\mathbf{g}^\bullet}\boldsymbol{F}_\nu^\star \delta_\emptyset = \delta_\emptyset^\ast \boldsymbol{F}_\nu{\mathbf{Z}_\mathrm{f}^\otimes}^\star\mathbf{E}_\nu\mathbf{X}_t\mathbf{E}_\nu \mathbf{Z}_\mathrm{f}^\otimes \boldsymbol{F}_\nu^\star\delta_\emptyset,
\end{equation}
from which the result follows by the $\star$-homomorphic property of $\boldsymbol{F}_\nu$.

To see that $\texttt{P}^t_\nu\big[{\widetilde{\mathbf{G}}^{\odot}_t(\varrho\otimes\mathbf{I}^\otimes) {\widetilde{\mathbf{G}}_t^{\odot^\star}}}\big]$ resolves (\ref{bl}) one need only note that the Lorentz transformed generator $\boldsymbol{\upsilon}^\star_\nu\widetilde{\mathbf{G}}\boldsymbol{\upsilon}_\nu$ defines the row-operator $\mathbf{V}^\star_\nu=\xi^\star\boldsymbol{\upsilon}^\star_\nu\widetilde{\mathbf{G}}\boldsymbol{\upsilon}_\nu=( I,\sqrt{\nu}\widetilde{\mathrm{L}},\nu K)$ and $\mathbf{V}^\star_\nu\varrho(t)\mathbf{V}_\nu=\mathbf{V}^\star\varrho(t)\mathbf{V}\nu$.
\end{proof}
\begin{xrem}
Notice that  the Lorentz transform of the Lindbladian generator $\mathbf{G}$ is the map
\begin{equation}
\left[
  \begin{array}{ccc}
    I & -\mathrm{L}^\ast & -\tfrac{1}{2}\mathrm{L}^\ast\mathrm{L} \\
    0 & \mathrm{I} & \mathrm{L} \\
    0 &  0 & I \\
  \end{array}
\right]\mapsto\left[
  \begin{array}{ccc}
    I & -\sqrt{\nu}\mathrm{L}^\ast & -\tfrac{\nu}{2}\mathrm{L}^\ast\mathrm{L} \\
    0 & \mathrm{I} & \sqrt{\nu}\mathrm{L} \\
    0 &  0 & I \\
  \end{array}
\right],
\end{equation}
which indeed boosts the deterministic part of the differential increment by the factor $\nu$, but see that it also boosts the quantum stochastic increments of creation and annihilation, and thus a Wiener increment, by the factor $\sqrt\nu$.
\end{xrem}

To illustrate this we shall again consider the sequential interaction dynamics (\ref{49b}) which may be given by an underlying counting differential equation in $\mathbb{G}_\nu$ as
\begin{equation}
\mathrm{d}\Psi_t(\vartheta)=\mathbf{J}\Psi_t(\vartheta)\mathrm{d}n_t(\vartheta)\label{groove}
\end{equation}
where $\mathbf{J}=\boldsymbol{\upsilon}_\nu\mathbf{G}\boldsymbol{\upsilon}_\nu^\star-\mathbf{I}$ and $\mathbf{G}$ is given by (\ref{inG}).

\subsection{Stochastic Perturbations of Schr\"odinger Dynamics}
It has been mentioned that a $\star$-generator $\mathbf{G}$ of Lindbladian dynamics may be obtained as the central limit of a counting process of sequential measurement when the measurement frequency $\nu$ is allowed to diverge, giving rise to a continual observation of the quantum system. However, one may simply begin with a Lindbladian generator $\mathbf{G}$ and  introduce the notion of an intensity of the counting interaction in the Minkowski-Poisson space that describes the underlying dynamics. In this case we shall interpret the pseudo-measurement intensity $\nu$ as a coupling strength of the apparatus to quantum object as was similarly done in the context of the Hilbert space $\mathcal{H}$ in \cite{Be02} where $\lambda=\sqrt{\nu}$ was discussed as a parameter of  the object-apparatus coupling also describing the accuracy of measurement continual diffusion-type measurement.

Let $\{\mathbf{Y}^s_r\}$ be a two-parameter family of operators describing a dilated Hamiltonian dynamics of a quantum system whose projections $\mathrm{Y}^s_r=\boldsymbol{F}\mathbf{Y}^s_r\boldsymbol{F}^\star$ resolve the time-dependent Schr\"odinger equation
\begin{equation}
\mathrm{d}\mathrm{Y}^t_r+\mathrm{i}H(t)\mathrm{Y}^t_r\mathrm{d}t=0,\quad \mathrm{Y}^r_r=\mathrm{I}.
\end{equation}
In the Minkowski-Hilbert space $\mathbb{H}$ the chronological exponent of the underlying QSDE, $\mathrm{d}\mathbf{Y}^t_r+\mathbf{K}_0(t)\mathbf{Y}^t_r\mathrm{d}n_t=0$,
has the form
\begin{equation}
\mathbf{K}_0(t)=\left[
                \begin{array}{ccc}
                  0 & 0 & \mathrm{i}H(t) \\
                  0 & 0 & 0 \\
                  0 & 0 & 0 \\
                \end{array}
              \right].
\end{equation}
We shall suppose that the $\star$-unitary operators $\mathbf{Y}^s_r$ describe the evolution of a quantum system under investigation. The investigation is carried out by performing measurements (observations) with respect to an apparatus (observer). These measurements are described by the incremental operators $\mathbf{L}(t)$ on the Minkowski-Hilbert space $\mathbb{H}$ and we shall suppose that the interactions between the quantum system and the apparatus are $\star$-unitary and of the form
\begin{equation}
\mathbf{I}+\mathbf{L}(t)= \left[
                \begin{array}{ccc}
                  I & -\mathrm{L}(t)^\ast & -\tfrac{1}{2}\mathrm{L}(t)^\ast\mathrm{L}(t) \\
                  0 & \mathrm{I} & \mathrm{L}(t) \\
                  0 & 0 & I \\
                \end{array}
              \right].\label{426}
\end{equation}
The measurement process is a sequential interaction dynamics in $\mathbb{H}$ given by the sequential action of the operators (\ref{426}). Moreover, the measurement process shall be regarded as a Poisson process of intensity $\nu$ causing a sequence of perturbations of the quantum system described by the  boosted perturbation equation
\begin{equation}
\mathrm{d}\mathrm{U}_t+\mathrm{i}H(t)\mathrm{U}_t\mathrm{d}t=\Lambda\big( \mathbf{D},\nu\mathrm{d}t\big)\label{427}
\end{equation}
on $\mathcal{H}$, where $\mathbf{D}(t)=\mathbf{L}(t)\mathrm{T}_t$ and $\Lambda\big( \mathbf{D},\nu\mathrm{d}t\big)= \Lambda\big( \boldsymbol{\upsilon}_\nu^\star\mathbf{D}\boldsymbol{\upsilon}_\nu,\mathrm{d}t\big)$ is a Lorentz transform of the interaction perturbations parameterizing the coupling strength of the apparatus to the quantum system.

This resulting underlying dynamics is given in $\mathbb{H}$ by the generator $\mathbf{G}=\mathbf{I}+\boldsymbol{\upsilon}_\nu^\star\mathbf{L}\boldsymbol{\upsilon}_\nu -\mathbf{K}_0$ corresponding to the underlying QSDE
\begin{equation}
 \mathrm{d}\mathbf{U}_t(\vartheta)=\mathbf{J}_\nu(t)\mathbf{U}_t(\vartheta)\mathrm{d}n_t(\vartheta),\quad \mathbf{J}_\nu=\boldsymbol{\upsilon}_\nu^\star\mathbf{L}\boldsymbol{\upsilon}_\nu-\mathbf{K}_0,
\end{equation}
where the propagator $\mathbf{U}_t$ describes the combined evolution of the apparatus and the quantum system, defining the stochastic unitary propagator $\mathrm{U}_t=\boldsymbol{F}\mathbf{U}_t\boldsymbol{F}^\star$ on $\mathcal{H}$.
However, by virtue of the QS Duhamel principle we may write the solution of (\ref{427}) as the projection
\begin{equation}
\mathrm{U}_t=\mathrm{Y}_0^t+\boldsymbol{\xi}^\star_t\big(\dot{\mathbf{Y}}^t_{\cdot}\boldsymbol{\upsilon}^\star_\nu \mathbf{L} \boldsymbol{\upsilon}_\nu\dot{\mathbf{U}}_{\cdot}\big)\boldsymbol{\xi}_t
\end{equation}
of the single-sum process
\begin{equation}
\mathbf{U}_t(\vartheta)=\mathbf{Y}_0^t(\vartheta)+\sum_{z\in\vartheta^t}\mathbf{Y}^t_z(\vartheta) \boldsymbol{\upsilon}^\star_\nu \mathbf{L}(z) \boldsymbol{\upsilon}_\nu \mathbf{U}_z(\vartheta),
\end{equation}
where $\dot{\mathbf{Y}}^t_{\cdot}(z):=\mathrm{Y}^t_z\otimes\mathbf{I}(z)$ and $\dot{\mathbf{U}}_{\cdot}(z):= \mathrm{U}_z\otimes\mathbf{I}(z)$ are extensions of $\mathrm{Y}^t_z$ and $\mathrm{V}_z$ onto the space $\dot{\mathcal{H}}(z):=\boldsymbol{\xi}(z)\mathcal{H}\equiv \mathbb{K}(z)\otimes\mathcal{H}$. This is given on the initial state $\psi_0=\psi\otimes\delta_\emptyset\in\mathcal{H}$ as
\begin{equation}
\psi_t(\vartheta)=\psi_0(\vartheta)+\sum_{z\in\vartheta^t}\sqrt{\nu}Y^t_z\mathrm{L}(z)\psi_z(\vartheta\setminus z)- \frac{1}{2}\int_0^tY^t_z\mathrm{L}(z)^\ast\mathrm{L}(z)\psi_z(\vartheta)\nu\mathrm{d}z,
\end{equation}
resolving the quantum stochastic diffusion equation
\begin{equation}
 \mathrm{d}\psi_t+\big(\tfrac{\nu}{2}\mathrm{L}^\ast\mathrm{L}+\mathrm{i}H\big)\psi_t\mathrm{d}t =\sqrt{\nu}L^k\mathrm{d}w^t_k\psi_t,
\end{equation}
which is the diffusion wave-equation describing the continuous decoherence of the quantum system coupled to the apparatus by the coupling constants $\lambda=\sqrt{\nu}$, recovering the unperturbed Schr\"odinger dynamics in the special case when $\lambda=0$ corresponding to the absence of measurement.

\section{Conclusion}
It should now be apparent that, from principles of mathematics, deterministic dynamics may henceforth be regarded as the projection of a stochastic process of discrete interactions. That is to say that the deterministic differential calculus of Newton and Leibniz has a discrete dilation which appears to be canonical when studying differential calculus using algebra.

From the point of view of calculations one may see the full implications of such a dilation in the case of the more general stochastic calculus where, once again, there is a canonical dilation of stochastic processes to the simplest processes of all - a counting process. The importance of the matrix-algebra representation of (quantum) It\^o algebra is not only likened to the realization of $C^\ast$-algebras as matrix algebras on Hilbert spaces, but it also transforms the manner in which we may perceive time. Namely as something that may be generated from a fundamental temporal spin, a primitive `time particle' having only two degrees of freedom called \emph{past} and \emph{future} which are meant as two very basic concepts from which conventional notions of flows of time may be derived.

Throughout the last century there has been great advances in physics but there has been no obvious success in understanding time other than its relativity of passing and its non-separability from space. Unfortunately this does not present us with any indication of an underlying mechanical structure, a mechanism, that describes how time `works'. The intention here was to provide details of such a mechanism in order to give a broader picture of the workings of time. What has been revealed from working with the Belavkin Formalism is that in the context of measurement a differential time increment is a spontaneous interaction at a boundary connecting past and future. Such may be called a measurement in a very general way, meaning \emph{a $\star$-unitary interaction in a Minkowski-Hilbert space}.

From a more down-to-Earth point of view a measurement may be understood as an observation; meaning that something is observing and something is observed and the interaction between the observing and the observed is the observation. So that we now have solid grounds for a discussion of time as a consequence of observation. That is observation in a very general sense, and a primary, or fundamental, detection may serve as better terminology.

Understanding the underlying counting differential equation (\ref{15}) as a process of sequential measurements of frequency $\nu$ means that  the boosted Schr\"odinger equation (\ref{boosted}) may be understood as the expected dynamics of a quantum system that arises from a very primitive kind of observation, namely the indication of the production of past from future, and without such observations the quantum system simply does not evolve. The idea that evolution itself arises fundamentally from observation does not appear to have been considered in the doctrine of science before these investigations were carried out. % the author.
\subsection{Hyperbolic Noise}
Whatever be one's preferred choice of terminology does not change the fact that we are able to understand deterministic and stochastic dynamics as a marginal dynamics of a purely counting dynamics in a Minkowski-Hilbert space. Notice that of the three external degrees of freedom $\{-,\bullet,+\}$ only the $\bullet$ degree of freedom was referred to as noise in the work above. However, in the context of underlying counting differential equation in Minkowski-Hilbert space the remaining degrees of freedom $\{-,+\}$ are also describing noises. After all, (\ref{15}) is a stochastic differential equation and the noises are described by the temporal-spin vectors $\xi\in(\mathbb{C}^2,\boldsymbol{\eta})$. These are the very basic noises that generate evolution, and the second-quantization $\boldsymbol{F}^\star$ may be regarded as an environment of these `time particles'.

The other remarkable thing is the hyperbolic geometry of the space within which the temporal-spin vectors are defined. The vectors $\xi$ are null-vectors $\xi^\star\xi=0$ and they arose from a consideration of the Newton-Leibniz differential increments $\mathrm{d}t$. Notice that (\ref{15}) does give rise to two distinct notions of time, the first being the variable components of the chains $\vartheta$ and the second being a parameter for the underlying counting dynamics. One may argue that we have re-introduced that which we initially set out to understand, but the parameter of the counting process may still be ordered by virtue of the basic ingredient $\xi$.

It will also be interesting to study the pseudo-measurement process in more detail for the case when the Minkowski metric $\boldsymbol{\eta}$ assumes the more general pseudo-Riemannian form
\[
\boldsymbol{\eta}(\nu)=\left[
  \begin{array}{ccc}
    0 & 0 & \nu_- \\
    0 & \nu_\bullet & 0 \\
    \nu_+ & 0 & 0 \\
  \end{array}
\right]\equiv\boldsymbol{\nu}
\]
with $\nu_\bullet$ having a similar form to (\ref{curve}).
\section{Appendix}
\subsection{Guichardet-Fock Space}
 The Guichardet-Fock space $\mathcal{F}$ \cite{Gui72} used here is a Hilbert space of vector-functions over an element $\Delta$, of the $\sigma$-field over $\mathbb{R}$, having non-zero Lebesgue measure. The vector functions $\chi\in\mathcal{F}$ map finite chains $\vartheta\subset\Delta$ into a product of Hilbert spaces $\mathcal{K}(t_1)\otimes\mathcal{K}(t_2)\otimes\cdots\otimes\mathcal{K}(t_n)$ over the points $t_i\in\vartheta$, where we have taken $\mathcal{K}(t)=\mathfrak{k}$ for simplicity. Here $\vartheta=\{t_1<t_2<\ldots<t_n\}$ and $n(\vartheta)=|\vartheta|\equiv n$ is called the cardinality of the chain $\vartheta$, and we can denote by $\mathcal{X}_\Delta$ the space of such finite chains $\vartheta\subset\Delta$.

The space $\mathcal{F}$ may also be written as $\Gamma(\mathcal{K})$ where $\Gamma$ is called the \emph{second-quantization functor} and $\mathcal{K}$ is a Hilbert space of square integrable functions $\mathrm{k}:z\rightarrow\mathfrak{k}$, $z\in\Delta$. The second-quantization functor admits the decomposition
\begin{equation}
\mathcal{F}=\mathcal{F}^t\otimes\mathcal{F}_{[t}
\end{equation}
where $\mathcal{F}^t$ is the restriction of $\mathcal{F}$ to $[\min\Delta,t)\subset\Delta$ and  $\mathcal{F}_{[t}$ is the restriction of $\mathcal{F}$ to $\Delta\setminus[\min\Delta,t)$.
The norm of such functions $\mathrm{k}$ is given as
\begin{equation}
\|\mathrm{k}\|^2=\int_{\Delta}\|\mathrm{k}(z)\|^2_{\mathfrak{k}}\mathrm{d}z
\end{equation}
where $\mathrm{d}z$ is understood as the Lebesgue measure on $\Delta$ and $\|\mathrm{k}(z)\|_{\mathfrak{k}}$ is the norm on the Hilbert space $\mathfrak{k}$ of  the vector $\mathrm{k}(z)$. Generally the Hilbert spaces $\mathcal{K}(z)$ need not be the same at each $z\in\Delta$. On $\mathcal{F}$ the norm of functions $\chi$ is given as
\begin{equation}
\|\chi\|^2:=\int_{\mathcal{X}_{\Delta}}\|\chi(\vartheta)\|^2\mathrm{d}\vartheta\equiv \sum_{n=0}^\infty\int _{\mathcal{X}_n}\|\chi(\vartheta_n)\|^2\mathrm{d}\vartheta_n\label{norm}
\end{equation}
where $\vartheta_n$ is any chain in $\Delta$ having cardinality $|\vartheta_n|=n$ and $\mathcal{X}_n$ is the space of such $n$-chains. More specifically,
\begin{equation}
\int _{\mathcal{X}_n}\|\chi(\vartheta_n)\|^2\mathrm{d}\vartheta_n:=\underset{\inf{\Delta}<t_1<\ldots< t_n<\sup{\Delta}}{\int\cdots\int}\|\chi(t_1,t_2,\ldots,t_n)\|^2\mathrm{d}t_1\mathrm{d}t_2\cdots\mathrm{d}t_n
\end{equation}
and for the empty chain $\vartheta_0=\emptyset$ we have
\begin{equation}
\int _{\mathcal{X}_0}\|\chi(\emptyset)\|^2\mathrm{d}\emptyset:=\|\chi(\emptyset)\|^2
\end{equation}
which is just the squared magnitude of $\chi(\emptyset)\in\mathbb{C}$. Notice that functions $\chi\in\mathcal{F}$ may be extended, for example, onto symmetric or antisymmetric Fock space using the appropriate isomorphic extensions of the functions $\chi_n$, over $\mathcal{X}_n$, onto $\mathbb{R}^n$.

Notice that we have a special class of vectors that we shall denote by $\mathrm{k}^\otimes$. These vectors are called product vectors, coherent vectors, or exponential vectors. They are they maps $\mathrm{k}^\otimes:\vartheta\rightarrow \mathrm{k}(t_1)\otimes\cdots\otimes\mathrm{k}(t_n)$ having the product form. They are called exponential vectors because their norm has the exponential form
\begin{equation}
\|\mathrm{k}^\otimes\|^2=\int_{\mathcal{X}_{\Delta}}\prod_{z\in\vartheta}\|\mathrm{k}(z)\|^2\mathrm{d}\vartheta=\exp\big\{ \|\mathrm{k}\|^2\big\}
\end{equation}
and this defines the coherent vectors as $\mathrm{k}^\otimes \exp\{ -\tfrac{1}{2}\|\mathrm{k}\|^2\}$. The space $\mathcal{F}$ is defined as the completion of the $\mathbb{C}$-linear span of these product vectors with respect to the norm (\ref{norm}).
The Guichardet-Fock space $\mathcal{F}$ appears very naturally in stochastic dynamics because at any time $t$ in some interval $\Delta$ we would like to consider the reduction of a quantum system arising from its interaction with a measurement apparatus, randomly perturbing the quantum system at times in an arbitrary chain $\vartheta\subset\Delta$.

\subsection{Weyl Transforms and The Fock-Vacuum Vector} Amongst the product vectors $\mathrm{k}^\otimes$ is $0^\otimes$ corresponding to the case $\mathrm{k}=0$; which is indeed a vector in $\mathcal{K}$. This vector is also denoted by $\delta_\emptyset$ where \begin{equation} \delta_\emptyset(\vartheta)=1\;\;\textrm{if}\;\vartheta=\emptyset,\;\;\textrm{otherwise}\;0.
\end{equation}
Notice that it is already normalized, $\|\delta_\emptyset\|=\exp\{0\}=1$.

Any coherent vector $\mathrm{k}^\otimes \exp\{ -\tfrac{1}{2}\|\mathrm{k}\|^2\}$ may be obtained by a unitary transform $\mathrm{W}\delta_\emptyset$ of the vacuum vector, where the unitary operator $\mathrm{W}$ is called a Weyl transformation. It has the form
\begin{equation}
\mathrm{W}=\exp\big\{\mathrm{A}^\star\mathrm{k}-\mathrm{k}^\ast\mathrm{A}\big\}\equiv\mathrm{W}_\mathrm{k}\label{weyl1}
\end{equation}
where $\mathrm{A}^\star\mathrm{k}$ and $\mathrm{k}^\ast\mathrm{A}$ are respectively quantum stochastic single integrals of creation and annihilation given on vectors $\chi\in\mathcal{F}$ as
\begin{eqnarray}
\big[\mathrm{A}^\star\mathrm{k}\chi\big](\vartheta)&=\sum_{z\in\vartheta}\mathrm{k}(z)\chi(\vartheta\setminus z)\\
\big[\mathrm{k}^\ast\mathrm{A}\chi\big](\vartheta)&= \int_\Delta\mathrm{k}(z)^\ast\chi(\vartheta \sqcup z)\mathrm{d}z
\end{eqnarray}
but the Weyl transformation may be given most beautifully in the Minkowski-Fock space as $\mathbf{Z}^\otimes$ where $\mathbf{Z}^\otimes(\vartheta):=\otimes_{t\in\vartheta}\mathbf{Z}(t)$ and $\mathbf{Z}^\star=\mathbf{Z}^{-1}$.

\subsection{The Fock Representation of Poisson Space}
This is the inductive limit $\mathcal{G}_\nu=\cup_{t\in\Delta}\mathcal{G}_\nu^t$, where $\mathcal{G}_\nu^s\subset\mathcal{G}_\nu^t$ for all $s<t$ in $\Delta$.  $\mathcal{G}_\nu^t$ is the space of functions $\varphi$ that are square integrable with respect to the Poisson measure $\mathrm{d}\texttt{P}^t_\nu(\vartheta)=\nu^\otimes(\vartheta)\mathrm{d}\vartheta\exp\{-\nu t\}$ for finite chains $\vartheta\subset[0,t)$, where $\min\Delta=0$ for convenience and $\nu^\otimes(\vartheta)=\otimes_{t\in\vartheta}\nu(t)$, so  $\varphi\in\mathcal{G}_\nu^t$ if
\begin{equation}
\int_{\mathcal{X}_{[0,t)}}\|\varphi(\vartheta)\|^2\mathrm{d}\texttt{P}^t_\nu(\vartheta)<\infty.
\end{equation}
 The Poisson intensity $\nu$ is considered to be a strictly positive element of $L^\infty(\Delta)$, and the space $\mathcal{G}_\nu^t$ is a \emph{scaled} Hilbert space isomorphic to  $\mathcal{H}^t=\mathfrak{h}\otimes\mathcal{F}^t$ with respect to the isometry $I^t_\nu=\exp\{\tfrac{1}{2}\nu t\}{\tfrac{1}{\sqrt\nu}}^\otimes$
so that the vectors $\chi\in\mathcal{H}^t$ define vectors $\varphi\in\mathcal{G}^t_\nu$ as
\begin{equation}
\varphi=I^t_\nu\chi\equiv \exp\{\tfrac{1}{2}\nu t\}{\tfrac{1}{\sqrt\nu}}^\otimes\chi.
\end{equation}
 In the case when $\chi$ is a product vector $\exp\{-\frac{1}{2}\mathrm{g}^\ast\mathrm{g}\} \mathrm{g}^\otimes$  we may write $\chi=\mathrm{W}_{\mathrm{g}}\delta_\emptyset$, where $\mathrm{W}_\mathrm{g}$ is given by (\ref{weyl1}), and define isometries of the Fock vacuum \cite{Be02}  as
\begin{equation}
 \mathrm{f}^\otimes=I^t_\nu(\mathrm{f})\delta_\emptyset:= \exp\{\tfrac{1}{2}\nu t\}{\tfrac{1}{\sqrt\nu}}^\otimes\mathrm{W}_{\mathrm{g}}\delta_\emptyset,
\end{equation}
where $\mathrm{g}^\ast\mathrm{g}:\nu t$ and with $\mathrm{g}=\sqrt{\nu}\mathrm{f}$ and $\|\mathrm{f}(z)\|=1$.
This also gives us the more general vectors $\varphi=b^i\mathrm{f}^\otimes_i$ as the transformations $I_\nu^t\chi$ of the vectors $\chi=c^i\exp\{-\frac{1}{2}\nu_it\} \mathrm{g}_i^\otimes$ where $\sqrt{\nu_i}\mathrm{f}_i=\mathrm{g}_i$, $\|\mathrm{f}_i(z)\|=1$, and $b^i=c^i\exp\{-\frac{1}{2}(\nu_i-\nu)t\} \sqrt{\frac{\nu_i}{\nu}}^\otimes$.

\subsection{Minkowski-Fock Space}
The Minkowski Fock space $\mathbb{F}=\Gamma(\mathbb{K})$ contains $\mathcal{F}$ as a subspace. The embedding operator $\boldsymbol{F}^\star$ takes any function $\chi\in\mathcal{F}$ to
\begin{equation}
\chi\otimes\left[
                                      \begin{array}{c}
                                        0 \\
                                        1 \\
                                      \end{array}
                                    \right]^\otimes
\end{equation}
in $\boldsymbol{F}^\star\mathcal{F}$ which is an affine embedding of $\mathcal{F}$ into $\mathbb{F}$.
The second quantization functor admits the identification
\begin{equation}
\Gamma(\Bbbk)\otimes\Gamma(\mathcal{K})=\Gamma(\Bbbk\oplus\mathcal{K})
\end{equation}
where $\Bbbk\oplus\mathcal{K}\cong\mathbb{K}$ and $\Bbbk=L^1(\Delta)\oplus L^\infty(\Delta)$ equipped with Minkowski metric (\ref{eta}), allowing us to write the embedding of  product functions $\chi=\mathrm{k}^\otimes$ as
\begin{equation}
\boldsymbol{F}^\star\mathrm{k}^\otimes=\left[
                                      \begin{array}{c}
                                        0 \\
                                        \mathrm{k}\\
                                        1 \\
                                      \end{array}
                                    \right]^\otimes
\end{equation}
which is the canonical form \cite{Be92b}.
The Minkowski-Fock space $\mathbb{F}$ only differs from the Fock space $\mathcal{F}$ above in so far as the `norm' on $\mathbf{f}\in\mathbb{K}$ is a pseudo-norm given by the Minkowski metric (\ref{2004}) as $\|\mathbf{f}\|^2=\mathbf{f}^\star\mathbf{f}=\mathbf{f}^\ast\boldsymbol{\eta}\mathbf{f}$, and it satisfies a positivity condition called \emph{conditional positivity} \cite{Be92b} when $\Re(f_-^{}f_+^\ast)\geq0$ given that $\mathbf{f}^\star=(f_-,\mathrm{f}^\ast,f_+)$. The norm of vectors $\Psi\in\mathbb{F}$ is given by
\begin{equation}
\|\Psi\|^2:=\int_{\mathcal{X}_{\Delta}}\|\Psi(\vartheta)\|^2\mathrm{d}\vartheta\equiv\Psi^\star\Psi\in\mathbb{R}
\end{equation}
where $\|\Psi(\vartheta)\|^2=\Psi(\vartheta)^\star\Psi(\vartheta)$, and we may also consider the Minkowski-Fock representation of Poisson space, called Minkowski-Poisson space and denoted by $\mathbb{G}_\nu^t$ having the pseudo-norm
\begin{equation}
\|\Phi\|_\nu^2:=\int_{\mathcal{X}_{\Delta}}\|\Phi(\vartheta)\|^2\mathrm{d}\texttt{P}^t_\nu(\vartheta).
\end{equation}

It is one of the great realizations of the Belavkin Formalism that quantum stochastic processes represented in $\mathcal{H}=\mathfrak{h}\otimes\mathcal{F}$ may be dilated to block diagonal operators in $\mathfrak{h}\otimes\mathbb{F}$. In the case of the Weyl operator we have $\mathrm{W}_\mathrm{g}=\boldsymbol{F}\mathbf{W}_\mathrm{g}\boldsymbol{F}^\star$ where $\mathbf{W}_\mathrm{g}=\mathbf{Z}_\mathrm{g}^\otimes$ is just a block-diagonal product-operator $\big[\mathbf{Z}_\mathrm{g}^\otimes\chi\big](\vartheta)= \Big(\mathbf{Z}_\mathrm{g}(t_1)\otimes\cdots\otimes\mathbf{Z}_\mathrm{g}(t_n)\Big)\chi(t_1,\ldots,t_n)$, $\vartheta=\{t_1<\ldots<t_n\}$, and this $\mathbf{Z}_\mathrm{g}$ is the $\star$-unitary operator
\begin{equation}
\mathbf{Z}_\mathrm{g}=\left[
  \begin{array}{ccc}
    1 & -\mathrm{g} & -\frac{1}{2}\mathrm{g}^\ast\mathrm{g} \\
    0 & \mathrm{I} & \mathrm{g} \\
    0 & 0 & 1 \\
  \end{array}
\right]\label{Z}
\end{equation}
and notice that this is very similar to the form of the generator of Lindblad dynamics, only without any action in $\mathfrak{h}$, referred to as the quantum-object Hilbert space.

\subsection{Quantum Stochastic Integration}
A QS (quantum stochastic) single-integral may be given by the action of an $\mathcal{B}(\mathcal{H})$-valued linear functional $\rho_\star$ on a quantum It\^o algebra $\mathfrak{a}(\mathcal{K})\otimes\mathcal{B}(\mathcal{H})$, where $\mathcal{B}(\mathcal{H})$ is an algebra of bounded operators on $\mathcal{H}=\mathfrak{h}\otimes\Gamma(\mathcal{K})$. The functional $\rho_\star$ is given on QS derivatives $\mathbf{D}$ as
\begin{equation}
\rho_\star(\mathbf{D})=\boldsymbol{\xi}^\star\mathbf{D}\boldsymbol{\xi},
\end{equation}
and $\mathbf{D}$ may also be referred to as the single-integral kernel.
A multiple-integral kernel $\mathbf{M}$ is an operator in the quantum It\^o algebra
$\Gamma\big(\mathfrak{a}(\mathcal{K})\big)\otimes\mathcal{B}(\mathcal{H})$ on $\mathbb{F}\otimes\mathcal{H}$, and a quantum stochastic multiple-integral is given on such  operators $\mathbf{M}$     as
\begin{equation}
 \rho_\star^\otimes(\mathbf{M})=\boldsymbol{\Xi}^\star\mathbf{M}\boldsymbol{\Xi}\label{38} \end{equation}
with respect to the operator $\boldsymbol{\Xi}:\mathcal{H}\rightarrow\Gamma(\mathbb{K})\otimes\mathcal{H}$ which may be given in terms of the vector-operator $\boldsymbol{\xi}$ as $\boldsymbol{\Xi}=\boldsymbol{\xi}^\otimes$ where $\boldsymbol{\xi}^\otimes=\big(\boldsymbol{F}^\star\otimes\mathrm{I}\big)\nabla^\otimes$ and $\nabla^\otimes:\mathcal{H}\rightarrow\mathcal{F}\otimes\mathcal{H}$ is the second-quantization of the point-derivative operator. The action of $\Xi$ on product functions $\mathrm{k}^\otimes\in\mathcal{F}$ is
\begin{equation}
\Xi\mathrm{k}^\otimes=\left[
                        \begin{array}{c}
                          0 \\
                          \mathrm{k} \\
                          1 \\
                        \end{array}
                      \right]^\otimes
\otimes\mathrm{k}^\otimes,
\end{equation}
and the multiple integral (\ref{38}) always defines a single integral as
\begin{equation}
\boldsymbol{\Xi}^\star\mathbf{M}\boldsymbol{\Xi}=\mathbf{M}(\emptyset)+\boldsymbol{\xi}^\star \boldsymbol{\Xi}^\star\dot{\mathbf{M}}\boldsymbol{\Xi} \boldsymbol{\xi}
\end{equation}
where $\dot{\mathbf{M}}$ is the point-derivative of $\mathbf{M}$ in the first (Minkowski-Fock space) variable. We may define generalized Maassen-Meyer kernels \cite{Be92b,LinM88a,Mey87} as multiple integral kernels of the form $\mathbf{M}=\mathbf{A}\otimes\mathrm{I}^\otimes$ defining the adapted quantum stochastic processes.
However, these kernels may be generalized without losing their adapted form, for we may define Q-adapted processes \cite{BelB12a,BelB12b,Brth} as those having the multiple integral kernels
\begin{equation}
\mathbf{M}=\mathbf{A}\otimes\mathrm{Q}^\otimes.\label{311} \end{equation}
In fact it was shown in \cite{BelB12b} that a Q-adapted process has a unique kernel of the form     (\ref{311}) up to the addition of a null kernel $\mathbf{Z}$ having $\boldsymbol{\Xi}^\star\mathbf{Z}\boldsymbol{\Xi}=0$.

The single-integral kernel of a Q-adapted process is also Q-adapted, and the
Q-adapted calculus gives a correction to the classical It\^o formula (coinciding with the adapted quantum It\^o formula) so that if $\mathbf{Q}:=\mathbf{I}  +\mathbf{L}$ then the quantum It\^o formula for two Q-adapted processes $\mathrm{X}$ and $\mathrm{Y}$ is
\begin{equation}
\mathrm{d}\big(\mathrm{X}\mathrm{Y}\big)=\Big(\mathrm{X}\mathrm{d}\mathrm{Y}+ \big(\mathrm{d}\mathrm{X}\big)\mathrm{d}\mathrm{Y} +\big(\mathrm{d}\mathrm{X}\big)  \mathrm{Y}\Big  ) +\Big(\mathrm{X}\mathrm{d}\Lambda(\mathbf{L})\mathrm{d}\mathrm{Y}+ \big(\mathrm{d}\mathrm{X}\big)\mathrm{d}\Lambda(\mathbf{L})\mathrm{Y}\Big).
\end{equation}

\subsection{The Lorentz Representation  of Minkowski Space}
As it turns out the complex Minkowski space $(\mathbb{C}^2,\boldsymbol{\eta})$ that we are using here may be called the \emph{Lorentz Representation}. This conclusion is fairly straightforward. As a self-adjoint operator on $\mathbb{C}^2$, $\boldsymbol{\eta}^\ast=\boldsymbol{\eta}$, the pseudo-metric (\ref{eta}) may be diagonalized by a unitary operator $\mathbf{U}$ to give
\begin{equation}
\mathbf{U}^\ast\boldsymbol{\eta}\mathbf{U}=\left[
                                   \begin{array}{cc}
                                     1 & 0 \\
                                     0 & -1 \\
                                   \end{array}
                                 \right],\qquad  \mathbf{U}=\tfrac{1}{\sqrt2}\left[
                 \begin{array}{cc}
                   1 & 1 \\
                   1 & -1 \\
                 \end{array}
               \right]
\end{equation}
 which is the time-like Minkowski metric familiar in special relativity. However, what we also find is that the real Lorentz transform is diagonalized by the same unitary operator $\mathbf{U}$ so that
 \begin{equation}
\mathbf{U}\left[
            \begin{array}{cc}
              \cosh\theta & \sinh\theta \\
              \sinh\theta & \cosh\theta \\
            \end{array}
          \right]\mathbf{U}^\ast=\left[
                                   \begin{array}{cc}
                                     \exp\{\theta\} & 0 \\
                                     0 & \exp\{-\theta\} \\
                                   \end{array}
                                 \right]\equiv \boldsymbol{\upsilon}^\star,\label{120}
\end{equation}
noting that we use $\sqrt\nu$ in place of $\exp\{\theta\}$ in this article,
and so we call $(\mathbb{C}^2,\boldsymbol{\eta})$ the Lorentz representation of the complex Minkowski space.
One may also wish to note that the future and past vectors also transform as
\[
\mathbf{U}^\dag\xi=\frac{1}{\sqrt2}\left[
                 \begin{array}{c}
                   1 \\
                   -1 \\
                 \end{array}
               \right],\quad \mathbf{U}^\dag\pi(\mathrm{d}t)\xi=\frac{1}{\sqrt2}\left[
                 \begin{array}{c}
                   1 \\
                   1 \\
                 \end{array}
               \right].
\]
In the case when we have the $\bullet$ degree of freedom of the noise we may use the Lorentz transform  to change the amplitude of vectors in the apparatus Hilbert space $\mathfrak{k}$. To understand this one should first recall that our basis vectors $\chi_{\mathrm{f}}:=\mathrm{f}\exp\{-\frac{1}{2}\mathrm{f}^\ast\mathrm{f}\}$ are generated by Weyl transforms of the vacuum vector $\delta_\emptyset$ so that $\chi_{\mathrm{f}}=\mathrm{W}_{\mathrm{f}}\delta_\emptyset$. The next thing to note is that $\mathrm{W}_{\mathrm{f}}=\boldsymbol{F}\mathbf{Z}^\otimes_{\mathrm{f}}\boldsymbol{F}^\star$. So now we shall consider a Lorentz transform $\boldsymbol{\upsilon}_\nu^\star\mathbf{Z}_{\mathrm{f}}\boldsymbol{\upsilon}_\nu$, given by (\ref{lorentz3}), of the generator $\mathbf{Z}_{\mathrm{f}}$ of the Weyl transform. When this is projected into $\mathcal{H}$ we obtain a new Weyl transform $\mathrm{W}_{\mathrm{g}}:=\boldsymbol{F}\big(\boldsymbol{\upsilon}_\nu^\star \mathbf{Z}_{\mathrm{f}}\boldsymbol{\upsilon}_\nu\big)^\otimes\boldsymbol{F}^\star$, where $\mathrm{g}=\sqrt{\nu}\mathrm{f}$, which is also $\boldsymbol{F}_\nu\mathbf{Z}^\otimes_{\mathrm{f}}\boldsymbol{F}_\nu^\star$. This gives us a new product vector $\chi_{\mathrm{g}}$ which is normalized, like $\chi_{\mathrm{f}}$, but with new magnitude in $\mathfrak{k}$.


\begin{thebibliography}{HD}
\bibitem[1]{Be88} V. P. Belavkin,
\emph{A New Form and a $\star$-Algebraic Structure of Quantum Stochastic Integrals in Fock Space},
Rediconti del Sem. Mat. e Fis. di Milano, LVIII, 1988, 177--193.

\bibitem[2]{Be92} V. P. Belavkin,
\emph{Quantum Stochastic Calculus and Quantum Non-Linear Filtering},
J. Multivariate Analysis, 1992, 42:171--201.

\bibitem[3]{Be92b} V.~P.~Belavkin,
\emph{Chaotic States and Stochastic Integration in Quantum Systems},
Russian Math. Surveys, 1992, 47:53--116.

\bibitem[4]{HudP84} R. L. Hudson and K. R. Parathasarathy,
\emph{Quantum It\^o Formula and Stochastic Evolutions},
Communications in Mathematical Physics, 1984.

\bibitem[5]{BelB12a} V.~P.~Belavkin and M.~F. Brown,
  \emph{Q-adapted quantum stochastic integrals and differentials in Fock scale}, Volume 96, {Noncommutative Harmonic Analysis with Applications to Probability III}, Banach Center Publications, 2012,  51--66.

\bibitem[6]{BelB12b} V.~P.~Belavkin and M.~F. Brown,
  \emph{$\mathrm{Q}
$-Adapted Integrals and It\^{o} Formula of Noncommutative Stochastic
Calculus in Fock Space}, {Communications on Stochastic Analysis (KRP Volume)}, 2012,  6:157--175.

\bibitem[7]{HudP84b} R. L. Hudson and K. R. Parthasarathy,
\emph{Unification of Fermion and Boson Stochastic Calculus},
Comm. Math. Phys, 1984, 93:301--323.

\bibitem[8]{Brth} M.~F. Brown,
  \emph{An Investigation of The Stochastic Representation of Quantum Evolution}, {PhD Thesis}, University of Nottingham, 2013.

\bibitem[9]{Br12} M.~F. Brown,
  \emph{The Stochastic Representation of Hamiltonian Dynamics and The Quantization of Time}, {arXiv:1111.7043}, 2012.

\bibitem[10]{Be00b} V.~P. ~Belavkin,
\emph{Quantum Trajectories, State Diffusion, and Time Asymmetric Eventum Mechanics},
2000.

\bibitem[11]{BelK01} V.~P. ~Belavkin and V.~N. ~Kolokoltsov,
\emph{Stochastic Evolutions as Boundary Value Problems}, Infinite Dimensional Analysis and Quantum Probability, Research Institute for Mathematical Studies, Kokyuroku, 2001, 1227:83-95.

\bibitem[12]{Be00a} V.~P.~Belavkin,
\emph{Quantum Stochastics, Dirac Boundary Value Problem, and The Ultra Relativistic Limit},
Rep. on Math. Phys, 2000, 46:359--382.

\bibitem[13]{Mal78} P. Malliavin,
\emph{Stochastic Calculus of Variations and Hypoelliptic Operators},
Proceedings of the International Symposium on Stochastic Differential Equantions, New York, 1978, 195--293.

\bibitem[14]{Sko75} A. V. Skorokhod,
\emph{On A Generalization of The Stochastic Integral},
Vol.~20, Theory Probab. Appl., 1975, 219--233.

\bibitem[15]{Be91} V. P. Belavkin,
\emph{A Quantum Non-Adapted It\^o Formula and Stochastic Analysis in Fock Scale},
J. Funct. Anal., 1991, 102:414--447.

\bibitem[16]{Be02} V.~P.~Belavkin,
\emph{Quantum Causality, Decoherence, Trajectories and Information}, Rept. on Prog. Phys., 2002, 65:353-420.

\bibitem[17]{Be94} V.~P.~Belavkin,
\emph{Non-Demolition Principle of Quantum Measurement Theory}, Foundations of Physics, 1994, 24:685-714.

\bibitem[18]{Be93} V.~P. ~Belavkin,
\emph{A Dynamical Theory of Quantum Measurement and Spontaneous Measurement},
Russian J.~Math. Phys, 1995, 3:3--24.

\bibitem[19]{BelM96} V.~P. ~Belavkin and O. Melsheimer,
\emph{A Stochastic Hamiltonian Approach for Quantum Jumps, Spontaneous Localizations, and Continuous Trajectories}, Quantum and Semi-Classical Optics, 1996, 8:167--187.

\bibitem[20]{Gui72} A. Guichardet,
\emph{Symmetric Hilbert Spaces and Related Topics}, Springer-Verlag, 1972.

\bibitem[21]{LinM88a} J. M. Lindsay and H. Maassen,
\emph{The Stochastic Calculus of Bose Noise},
Preprint, 1988.

\bibitem[22]{Mey87} P. A. Meyer,
\emph{\'El\'ements de Probabilit\'es Quantiques},
Vol.~1247, Lecture Notes in Mathematics, 1987, 33--78.
\end{thebibliography}
\end{document}